\documentclass[12pt, reqno]{amsart}

\usepackage{graphics, stackrel}
\usepackage{amsmath, amssymb, amsthm}
\usepackage{graphicx}
\usepackage{verbatim}
\usepackage{amsfonts}

\usepackage{natbib}

\usepackage{enumitem}
\usepackage[T1]{fontenc}
\usepackage{mathpazo}

\usepackage{caption}
\usepackage{subcaption}

\usepackage{booktabs}

\usepackage{fancyvrb}
\usepackage[usenames, dvipsnames, svgnames,table]{xcolor}
\usepackage{mdwlist}

\usepackage{enumitem}
\setlist[enumerate]{itemsep=2pt,topsep=3pt}
\setlist[itemize]{itemsep=2pt,topsep=3pt}
\setlist[enumerate,1]{label=(\alph*)}

\usepackage{mathrsfs}
\usepackage{bbm}

\usepackage[left=1.25in, right=1.25in, top=1.0in, bottom=1.15in, includehead, includefoot]{geometry}
\usepackage[citecolor=blue, colorlinks=true, urlcolor=blue, linkcolor=blue]{hyperref}

\renewcommand{\leq}{\leqslant}
\renewcommand{\geq}{\geqslant}

\usepackage[ruled, linesnumbered]{algorithm2e}

\newcommand\blfootnote[1]{%
  \begingroup
  \renewcommand\thefootnote{}\footnote{#1}%
  \addtocounter{footnote}{-1}%
  \endgroup
}

\newcommand{\setntn}[2]{ \{ #1 : #2 \} }

\newcommand{\iidsim}{\stackrel {\textrm{ {\sc iid }}} {\sim} }
\newcommand{\1}{\mathbbm 1}
\newcommand{\la}{\langle}
\newcommand{\ra}{\rangle}

\newcommand*\diff{\mathop{}\!\mathrm{d}}

\newcommand{\rR}{\mathcal R}

\newcommand{\given}{\,|\,}

\newcommand{\bB}{\mathscr B}

\newcommand{\hH}{\mathcal H}

\newcommand{\iI}{\mathcal I}

\newcommand{\eE}{\mathcal E}
\newcommand{\fF}{\mathscr F}

\newcommand{\lL}{\mathcal L}

\newcommand{\RR}{\mathbbm R}

\newcommand{\NN}{\mathbbm N}

\newcommand{\PP}{\mathbbm P}

\newcommand{\XX}{\mathsf X}
\newcommand{\WW}{\mathsf W}

\newcommand{\EE}{\mathbbm E}

\theoremstyle{plain}
\newtheorem{theorem}{Theorem}[section]

\newtheorem{lemma}[theorem]{Lemma}
\newtheorem{proposition}[theorem]{Proposition}

\theoremstyle{definition}

\newtheorem{remark}{Remark}[section]

\newtheorem{assumption}{Assumption}[section]

\usepackage[normalem]{ulem}

\begin{document}


\title{}

\begin{center}
  \LARGE
  Stability of Equilibrium Asset Pricing Models: \\ A Necessary and Sufficient
  Condition\blfootnote{We thank Ippei Fujiwara, Sean Meyn,
      Thomas J.\ Sargent, co-editor Guillermo Ordo\~{n}ez and two referees for
      many valuable comments and suggestions, as well as
      participants at the plenary lecture of the SNDE 2018 Annual symposium at
      Keio University.  In addition, special thanks are due to Miros\l{}awa
  Zima for detailed discussion of local spectral radius conditions.  Financial
  support from ARC grant FT160100423 and Alfred P. Sloan Foundation grant
  G-2016-7052 is also gratefully acknowledged.}

    \vspace{1em}

  \large
  Jaroslav Borovi\v{c}ka\textsuperscript{a}
  and John Stachurski\textsuperscript{b} \par \bigskip

  \small
  \textsuperscript{a} New York University and NBER \par
  \textsuperscript{b} Research School of Economics, Australian National University \bigskip

  \normalsize
  \today
\end{center}

\begin{abstract}
    We obtain an exact necessary and sufficient condition for the existence and uniqueness of equilibrium asset prices in infinite horizon, discrete-time, arbitrage free environments.  Using local spectral radius methods, we connect the condition, and hence the problem of existence and uniqueness of asset prices, with the recent literature on stochastic discount factor decompositions.  Our results include a globally convergent method for computing prices whenever they exist.  Convergence of this iterative method itself implies both existence and uniqueness of equilibrium asset prices.
    \vspace{1em}

    \noindent
    \textit{JEL Classifications:} D81, G11 \\
    \textit{Keywords:} Asset pricing, equilibrium prices, spectral methods
\end{abstract}

\maketitle

\section{Introduction}

One fundamental problem in economics is the pricing of an asset
paying a stochastic cash flow with no natural termination point, such as
a sequence of dividends.  In discrete-time no-arbitrage environments, the
equilibrium price process $\{P_t\}_{t \geq 0}$ associated with a
dividend process $\{D_t\}_{t \geq 1}$ obeys
\begin{equation}
    \label{eq:pd}
    P_t
    = \EE_t \, M_{t+1} ( P_{t+1} + D_{t+1} )
    \quad \text{for all } t \geq 0,
\end{equation}
where $\{ M_t\}$ is the sequence of single period stochastic
discount factors.\footnote{See, for example,~\cite{kreps1981arbitrage},
    \cite{hansen_richard:1987} or~\cite{duffie2001dynamic}.  Here and below,
    prices are on ex-dividend contracts.  (Cum-dividend contracts are a simple
extension.)} Two questions immediately
arise in connection with
these dynamics:
\begin{enumerate}
    \item[1.] Given $\{D_t, M_t\}_{t \geq 1}$, does there exist a unique
        equilibrium price process?
    \item[2.] How can we characterize and evaluate such prices
        whenever they exist?
\end{enumerate}

These questions have become more pressing for two reasons.  First, models of
dividend processes and state price deflators are becoming more sophisticated,
in an ongoing effort to better match financial data and resolve outstanding
puzzles (see, e.g., recent iterations of the models in
\cite{campbell1999force}, \cite{barro:2006}, or \cite{bansal2004risks}).  This
complexity makes questions 1--2 challenging, especially in quantitative
applications with discount rates close to the growth rates of underlying cash
flows.  There have been few sufficient conditions proposed that (a) imply
existence and uniqueness of equilibrium prices, (b) are weak enough to be
useful in modern quantitative analysis, and (c) are practical enough to
implement in interesting applied settings.

The second reason that questions 1--2 have become more pressing is the
accumulating evidence that nonlinearities embedded in the original models
matter for quantitative analysis.  For example,~\cite{pohl2018higher} and
\cite{lorenz2020nonlinear} show that the log-linearization techniques used to
solve asset pricing models can lead to large distortions in the equity premium
and price volatility.  These findings increase the need for practical methods
for investigating the underlying structure of modern asset pricing models.

In this paper, we introduce a condition for existence and uniqueness of equilibria
that is both weak enough to hold in realistic applications---in fact necessary
as well as sufficient---and practical in the sense that testing the condition
focuses on a single value.  This value is referred to below as the
\emph{stability exponent}.
To illustrate key ideas, consider the case of stationary dividend growth,
which is the standard assumption in quantitative applications.  Seeking a
stationary price-dividend ratio, we rewrite (\ref{eq:pd}) as
\begin{equation}
    \label{eq:pd2}
    \frac{P_t}{D_t}
       = \EE_t \left[
                 M_{t+1} \frac{D_{t+1}}{D_t}
                 \left(\frac{P_{t+1}}{D_{t+1}} + 1\right)
               \right].
\end{equation}

Let $\Phi_{t+1}:=M_{t+1}\left(D_{t+1}/D_t\right)$.
For this class of models, the stability exponent is
\begin{equation}
    \label{eq:kr_nonstationary}
    \lL_\Phi :=
        \lim_{n \to \infty}
        \frac{\ln \psi_n}{n}  ,
\end{equation}
where $\psi_n := \EE \prod_{t=0}^{n-1} \Phi_{t+1}$ is the expectation of the
$n$-period pricing kernel adjusted for dividend growth.  Uncertainty in the
discount process $\{M_t\}$ and dividend growth
$\{D_{t+1}/D_t\}$ is driven by an irreducible state process, and $\EE$ takes
expectations over the unique stationary distribution.

We show that, in this setting, existence and uniqueness of an equilibrium
price process  is exactly equivalent to the statement $\lL_\Phi < 0$.  In
addition, successive approximations converge globally to equilibrium prices if
and only if $\lL_\Phi < 0$.  We also show that convergence of this algorithm
itself implies that the limit is an equilibrium and, moreover, that $\lL_\Phi
< 0$.  Therefore, convergence from a single initial condition implies that the
limit is an equilibrium, and that these equilibrium prices are unique and globally
attracting.


Interpreting the condition $\lL_\Phi < 0$ is straightforward.
Let $p_n(x) := \EE_x \, \prod_{t=0}^{n-1}  \Phi_{t+1}$ denote the price of a claim
on the dividend paid out $n$ periods ahead at the current state $x$,
normalized by the current dividend. The pricing of these so-called dividend
strips has been analyzed extensively, see, e.g.,
\cite{binsbergen_brandt_koijen:2012} or
\cite{binsbergen_hueskes_koijen_vrugt:2013}. Due to irreducibility, $\EE_x$
can be replaced by $\EE$ for limiting events, so $p_n(x) \approx \psi_n$ for
large $n$. The condition then states that, asymptotically, prices of long-horizon dividend
strips
decay to zero at a geometric rate.\footnote{From~\eqref{eq:kr_nonstationary} and $p_n(x) \approx
    \psi_n$ we have $p_n(x)^{1/n} \approx \exp(\lL_\Phi)$ for large $n$, so
    the dividend strip price $p_n(x)$ goes to zero like $\exp(n \lL_\Phi)$ when $\lL_\Phi < 0$. The value
$-\lL_\Phi$ is the decay rate.}

The intuition is particularly simple in the case where
dividends are stationary. We can then replace $\Phi_{t+1}$ in
    (\ref{eq:kr_nonstationary}) with $M_{t+1}$.  Now $p_n(x) :=
    \EE_x \, \prod_{t=0}^{n-1}  M_{t+1}$ is the price of a risk-free
    zero-coupon bond with maturity $n$ at current state $x$, so $y_n(x) := -
    \ln p_n(x) / n$ is the corresponding yield to maturity.
    Since $p_n(x) \approx \psi_n$ for large $n$, the condition $\lL_\Phi <
    0$ means that, in the limit, yields on risk-free long bonds are positive.
This indicates a fundamental preference for current payoffs over future
payoffs,  which generates finite, well defined prices for stationary infinite horizon
cash flows.




While $\lL_\Phi < 0$ has a
natural interpretation, it is striking that this condition
is necessary as well as sufficient for existence, and hence exactly
characterizes the set of models with well defined equilibrium prices.  This
result rests on irreducibility mentioned above, which is a mild condition, and
a  ``local spectral radius'' result due to \cite{zabreiko1967bounds} and
\cite{forster1991local}. Using this result, we show that, for any positive
cash flow with finite first moment, the asymptotic mean growth rate of its
discounted payoff stream is equal to the principal eigenvalue of an associated
valuation operator, which is in turn equal to the exponential of $\lL_\Phi$.
If the principal eigenvalue equals or exceeds unity, then the sum of expected
discounted payoffs grows without bound.

An operator-theoretic way to understand our results is to view $\Phi_{t+1}$ as
a random ``contraction factor'' for an operator that has, as its fixed point,
an equilibrium price function.  If there is a constant $\theta$ with
$\Phi_{t+1} \leq \theta < 1$ with probability one, then valuation equations
\eqref{eq:pd} and \eqref{eq:pd2} imply this operator will be a contraction of
modulus $\theta$, yielding existence of a unique equilibrium.  However, in
most applications, $\Phi_{t+1} > 1$ holds on a set of positive probability,
due to the fact that payoffs in bad states have high value.  Thus, a direct
one step contraction argument is problematic.  Hence we adopt the weaker
condition $\lL_\Phi < 0$, which requires instead that $\Phi_{t+1} < 1$ holds
on average over the long run, and show that this is both necessary and
sufficient.\footnote{The growth rate $\lL_\Phi$ is also connected to the
    integrated Lyapunov exponent (see, e.g., \cite{knill1992positive}), which,
    for the process $\{\Phi_t\}$, takes the form $\iI_\Phi := \lim_{n \to
    \infty} \frac{1}{n} \, \sum_{t=1}^n \EE \, \ln \Phi_t$.  When $\{\Phi_t\}$
    is stationary, this reduces to $\EE \, \ln \Phi_t$.  Jensen's inequality
    yields $\iI_\Phi \leq \lL_\Phi$, and, as $\lL_\Phi < 0$ is necessary and
    sufficient for existence and uniqueness, the inequality $\iI_\Phi < 0$ is
    necessary but not sufficient.  This is because it takes into account only
    the marginal distribution of $\Phi_t$, and hence ignores long epochs
    during which the SDF exceeds unity.  Stability requires controlling
    persistence in the SDF process, which requires restrictions on the full
joint distribution.}

We also discuss methods for calculating $\lL_\Phi$ when no analytical solution
exists.  Similar to \cite{backus_gregory_zin:1989}, we show that, when the
state space is finite, the rate of decay of prices of long-term dividend strips and hence
$\lL_\Phi$ can be calculated using numerical linear algebra.  For other cases,
we propose a Monte Carlo method that involves simulating independent
realizations of the pricing kernel.  This method is inherently parallelizable,
sufficiently accurate for the applications we consider, and well suited to
settings where the state space is large.\footnote{Appendix~\ref{s:ctv}
    provides details.  Overall, we find that modern quantitative asset pricing
    studies are too complex---and too close to the boundary between stability
    and instability---to allow for successful use of purely analytical
sufficient conditions.}

As one illustration of the method, we consider a model of asset prices with
Epstein--Zin recursive utility, multivariate cash flows and time varying
volatility studied in \cite{schorfheide2018identifying}.  Hitherto no results
have been available on existence and uniqueness of equilibria in the
underlying theoretical model. We show that $\lL_\Phi < 0$ holds at and in the
neighborhood of the benchmark parameterization in a global
numerical approximation of the model.   This indicates
existence of a unique set of equilibrium prices, along with a globally
convergent method of computing them.  The fact that our conditions are
necessary as well as sufficient allows us to examine how far this positive
result can be pushed as we shift parameters relative to the benchmark.


We also encompass and extend the classic result of
\cite{lucas1978asset}, who studied a model with infinite state space and SDF
of the form
\begin{equation}
    \label{eq:ms}
    M_{t+1} = \beta \frac{u'(C_{t+1})}{u'(C_t)}.
\end{equation}
Here $\{C_t\}$ is a stationary consumption process, $\beta \in (0, 1)$ is a
state independent discount component and $u$ is a period utility function.
Using a change of variable, \cite{lucas1978asset} obtains a modified pricing
operator with contraction modulus equal to $\beta$, and hence, by Banach's
contraction mapping principle, a unique equilibrium price process.  His
theorem is a special case of our main result.

While \cite{lucas1978asset} frames his study in a space of bounded functions,
our analysis admits unbounded solutions.  This is achieved by embedding the
equilibrium problem in a space of candidate solutions with finite first
moments.  Such a setting is arguably more natural for the study of forward
looking stochastic sequences, since the forward looking restriction is itself
stated in terms of expectations.  Adopting this setting allows us to
generalize the existence and uniqueness results for equilibrium prices
obtained in \cite{calin2005solving} and \cite{brogueira2017existence}, which
extend \cite{lucas1978asset} by allowing for habit formation and unbounded
utility.\footnote{Not surprisingly, our results also generalize the simple
    risk neutral case $M_t \equiv \beta$, which is linear and hence easily
    treated by standard methods (see, e.g., \cite{blanchard1980solution}).
    The existence of a unique equilibrium when $\beta \in (0,1)$ is a special
    case of our results because the $n$ period state price deflator is just
    $\beta^n$, so, by the definition of $\lL_\Phi$ in
    \eqref{eq:kr_nonstationary} and with stationary consumption, we have
    $\lL_\Phi=\ln \beta$. The condition $\beta \in (0,1)$ therefore
    implies $\lL_\Phi < 0$ and hence existence of a unique solution.}

Our work is also connected to the literature on stochastic discount factor
decompositions found in \cite{alvarez2005, hansen2009long, hansen2012dynamic,
borovivcka2016misspecified, christensen2017nonparametric, qin2017long} and
other recent studies.  These decompositions are used to extract a permanent
growth component and a martingale component from the stochastic discount
process, with the rate in the permanent growth component being driven by the
principal eigenvalue of the valuation operator associated with stochastic
discount factor.
We show that the log of this principal eigenvalue is equal to $\lL_\Phi$ in our
setting, using the local spectral radius result discussed above.

In addition, our work is related to
\cite{pohl2019rel} and \cite{christensen2020existence}, who provide conditions for existence and
    uniqueness of recursive utilities in settings where the state space and
    rewards are unbounded.  While the objects in play are different
    (recursive utilities vs asset prices), the techniques are related because
    both sets of problems treat forward looking recursions over unbounded
    state spaces driven by exogenous state processes.  The connection can be summarized as follows: Our results
    are not applicable for most recursive utility problems concerning 
    existence and uniqueness, where
    nonlinearities require specialized techniques (e.g., the Orlicz space
    methods in \cite{christensen2020existence} or Jensen-type bounds in
    \cite{pohl2019rel}).  At the same time, our methods have
    comparative advantage for asset pricing, because they exploit
    positivity and affine structure (which
    follow from nonexistence of arbitrage). This allows us to obtain
    conditions for existence and uniqueness that are necessary as well as
    sufficient.\footnote{By the same token,
    our work is connected to \cite{borovivcka2020necessary}, which also
    treats existence and uniqueness of recursive utilities, with an emphasis on
    Epstein--Zin preferences.  As with \cite{christensen2020existence}, the
    topics differ but the focus on forward looking recursions is shared.
    Regarding a comparison of content, most of the same comments apply.
    (Orlicz space methods in \cite{christensen2020existence} are replaced by
    monotone concave operator methods in \cite{borovivcka2020necessary}.)
    However, \cite{borovivcka2020necessary} is less comparable
    to the present paper than \cite{christensen2020existence}, since, in
    the former, the focus is on compact states.}  In addition, we use this same no arbitrage structure,
    combined with properties of positive linear operators, to translate
    spectral radius conditions over valuation operators into the analytically
    and computationally convenient stability exponent $\lL_\Phi$.  Finally, we
offer techniques for computing $\lL_\Phi$ analytically, as well as via linear algebraic
methods and through Monte Carlo.

The rest of the paper is structured as follows. The main results are
presented in Section~\ref{s:mr}.  Sections~\ref{s:aar}--\ref{s:aii} treat
applications with stationary dividend growth and
Section~\ref{s:c} concludes.   Appendix~\ref{s:aiii} discusses models with
stationary dividends, rather than stationary dividend growth.  A
discussion of numerical methods for implementing our test can be found in
Appendix~\ref{s:ctv}.  Long proofs are deferred until Appendix~\ref{s:ax}.
Computer code that replicates our numerical results and figures can be found
at \url{https://github.com/jstac/asset_pricing_code}.

\section{A Necessary and Sufficient Condition}

\label{s:mr}

In this section we present our framework and state our main results.

\subsection{Environment}

\label{ss:msol}

We will work with the generic forward looking model
\begin{equation}
    \label{eq:fwl}
    Y_t = \EE_t \left[ \Phi_{t+1} (Y_{t+1} + G_{t+1} ) \right]
    \quad \text{for } t \geq 0,
\end{equation}
where $\{(\Phi_t, \, G_t)\}$ is a given stochastic process,
defined on some underlying probability space $(\Omega, \fF, \PP)$, and
$\{Y_t\}$ is endogenous.  While
other interpretations are possible,  it is convenient to refer to $\{\Phi_t\}$
as the \emph{stochastic discount factor} and $\{G_t\}$ as the \emph{cash
flow}.

Equations \eqref{eq:pd} and \eqref{eq:pd2} are special cases
of this recursion. In the latter case, the price dividend
ratio $P_t/D_t$ is the endogenous process, $\Phi_{t+1} = M_{t+1}D_{t+1}/D_t$
is a growth-adjusted stochastic discount factor, and the cash flow is
$G_{t+1}=1$.

We say that a stochastic process $\{Y_t\}$ \emph{solves} \eqref{eq:fwl} if,
with probability one, each $Y_t$ is finite and \eqref{eq:fwl} holds for all $t
\geq 0$.  To obtain a solution we require some auxiliary conditions on the
state process, the cash flow and the stochastic discount process.  The first
is as follows:

\begin{assumption}\label{a:p}
    For all $t \geq 0$, we have $\PP\{\Phi_t > 0\}=1$, while $G_t \geq 0$
    and $G_t > 0$ with positive probability.
\end{assumption}

Positivity of $\Phi_t$ is equivalent to
assuming no arbitrage (\cite{hansen_richard:1987}, Lemma~2.3) and holds in all
applications we consider. Provided that the focus is on nonnegative cash flows
(e.g., dividends), the second condition is also innocuous, since $G_t = 0$ almost
surely implies $Y_t=0$ for all $t$.

To introduce the possibility of stationary Markov solutions, we assume that
$\{\Phi_t\}$ and $\{G_t\}$ admit the representations
\begin{equation}
    \label{eq:azdef}
    \Phi_{t+1} = \phi(X_t, X_{t+1}, \eta_{t+1})
    \quad \text{and} \quad
    G_{t+1} = g(X_t, X_{t+1}, \eta_{t+1})
\end{equation}
where $\{X_t\}$ is an underlying $\XX$-valued state process,
$\{\eta_t\}$ is a $\WW$-valued innovation sequence  and $\phi$ and $g$
are positive Borel measurable maps on $\XX \times \XX \times \WW$. The sets
$\XX$ and $\WW$ can be any
separable and completely metrizable topological spaces.
The representations in \eqref{eq:azdef} replicate the multiplicative
functional specifications in \cite{hansen2009long} and
\cite{hansen2012dynamic}.

The innovation process $\{\eta_t\}$ is assumed to be {\sc iid}
and independent of $\{X_t\}$, with common distribution $\nu$.  The state
process is assumed to be stationary and Markovian with common
marginal distribution $\pi$.  The conditional
distribution of $X_{t+1}$ given $X_t = x$ is denoted by $\Pi(x, \diff y)$.
We use $\Pi^n$ to represent $n$-step transition probabilities (see, e.g.,
\cite{meyn2009markov}).

\begin{assumption}
    \label{a:i}
    The state process $\{X_t\}$ is irreducible: for each Borel set $B \subset
    \XX$ with $\pi(B) > 0$ and each $x \in \XX$, there exists an $n \in \NN$
    such that $\Pi^n(x, B) > 0$.
\end{assumption}

Assumption~\ref{a:i} is a weak mixing condition on the state process that is
satisfied in all applications we consider.  For example,
\cite{mehra1985equity} use a discrete state space and require that the Markov
chain is both irreducible and aperiodic.   Similarly, in the long run risk model of
\cite{schorfheide2018identifying}, innovations have positive densities, which
implies irreducibility.

A measurable function $h$ from $\XX$ to $\RR$ is called a \emph{Markov
solution} to \eqref{eq:fwl} if
\begin{equation*}
    h(X_t) = \EE_t \left[
        \Phi_{t+1} (h(X_{t+1}) + G_{t+1} )
        \right]
\end{equation*}
for all $t \geq 0$, which means that $\{Y_t\} := \{ h(X_t) \}$ solves \eqref{eq:fwl}.
Conditioning on $X_t = x$, we see that $h$ will be a Markov solution if
it is a fixed point of the \emph{equilibrium
price operator} $T$ defined by
\begin{equation}
    \label{eq:deff}
    T h (x)
    = \EE
    \left[
        \Phi_{t+1}
        \left(h(X_{t+1}) + G_{t+1}\right) \given X_t = x
    \right].
\end{equation}
For each $p \geq 0$, we let $L_p(\XX, \RR, \pi)$ denote, as
usual, the set of Borel measurable real-valued functions $h$ defined on the
state space $\XX$ such that $\int |h(x)|^p \pi(\diff x)$ is finite.  Let
$\hH_p$ be all nonnegative functions in $L_p(\XX, \RR, \pi)$.
  This will be our candidate space, so if $p=2$, say,
then we seek solutions with finite second moment.\footnote{In what follows,
all notions of convergence refer to standard norm convergence in
$L_p$.  As usual, functions equal $\pi$-almost everywhere are
identified.  Appendix~\ref{s:ax} gives more details.}


\begin{assumption}
    \label{a:co}
    There exists a $p \geq 1$ such that $\EE \, (\Phi_t G_t)^p < \infty$ and, in
    addition, the map $h \mapsto Vh$ defined by
    \begin{equation}
        \label{eq:aop}
        V h(x) =  \EE \left[ \Phi_{t+1} \, h(X_{t+1}) \given X_t = x \right]
    \end{equation}
    is eventually compact as a
    linear operator from $L_p(\XX, \RR, \pi)$ to itself.
\end{assumption}

In what follows, we call $V$ the \emph{valuation operator}.  The first part of
Assumption~\ref{a:co}, which requires that the cash flow process has $p$
finite moments after discounting, is weakest when $p=1$.  In fact, this
minimal restriction cannot be omitted, since \eqref{eq:fwl} is not well
defined without finiteness of first moments.\footnote{We might wish to choose
    $p$ to be larger when possible, in order to impose more structure on our
    solution (e.g., finiteness of second moments for asymptotic results
related to estimation).} The ``eventually compact'' part of
Assumption~\ref{a:co} is a regularity condition related to the notion of
compact linear operators, stated formally in
Appendix~\ref{s:ax}.  In Sections~\ref{s:aar}--\ref{s:aii} we discuss how to test this
condition and review its implications.\footnote{Analogous conditions can be
    found in the literature on eigenfunction decompositions of valuation
    operators.  See, for example, Assumption~2.1 in
\cite{christensen2017nonparametric}.}

\subsection{Existence and Uniqueness}

\label{ss:eu}

We introduce the \emph{$p$-th order stability exponent} of the SDF process
$\{\Phi_t\}$ as
\begin{equation}
    \label{eq:prpp}
    \lL_\Phi^p  :=
    \lim_{n \to \infty}
    \frac{1}{np} \ln \EE \left\{ \EE_x  \prod_{t=1}^n \, \Phi_t \right\}^p .
\end{equation}
Here and below, $\EE_x$ conditions on $X_0 = x$.
The simplest case is when $p=1$, since, by the Law of Iterated Expectations,
\begin{equation}
    \label{eq:prp}
    \lL_\Phi :=
    \lL_\Phi^1 =
    \lim_{n \to \infty}
    \frac{1}{n} \ln \left\{ \EE \prod_{t=1}^n \, \Phi_t \right\}.
\end{equation}
As discussed in the introduction, when $\{\Phi_t\}$ is the
discount factor process, $-\lL_\Phi$  can be interpreted as the yield on a
zero-coupon bond with very long maturity. In the case of a growth-adjusted
discount factor process from \eqref{eq:pd2}, $-\lL_\Phi$ is
the rate of decay of prices of long-maturity dividend strips.


\begin{theorem}
    \label{t:bk}
    If Assumptions~\ref{a:p}--\ref{a:co} hold, then
    the limit in \eqref{eq:prpp} exists and
    all of the following statements are equivalent:
    \begin{enumerate}
        \item $\lL_\Phi^p < 0$.
        \item There exist $h_0, h$ in $\hH_p$ such that $T^n h_0 \to h$
            as $n \to \infty$.
        \item A Markov solution $h^*$ exists in $\hH_p$.
        \item A unique Markov solution $h^*$ exists in $\hH_p$ and
            $T^n h  \to h^*$  for every $h \in
            \hH_p$.
    \end{enumerate}
    If one and hence all of (a)--(d) are true,
    then $h^*$ satisfies
    \begin{equation}
        \label{eq:fsg}
        h^*(x) =
        \sum_{n=1}^\infty
            \EE_x
            \,
            \prod_{i=1}^{n} \Phi_i \, G_n
        \quad \text{for $\pi$-almost all $x$ in $\XX$}.
    \end{equation}
\end{theorem}

The $p$ in conditions (a)--(d) is from Assumption~\ref{a:co}.
    Condition (b) is valuable from an applied perspective, since it
        shows that if iteration with $T$ converges from some starting
        point, then the limit is necessarily a Markov solution, and, in fact
        is the only Markov solution in $\hH_p$.
    Part (d) shows that successive approximations is
            globally convergent whenever $\lL_\Phi^p < 0$.\footnote{Successive
        approximations can be compared with other methods for solving asset
        pricing models, such as perturbations and projections (see, e.g.,
        \cite{pohl2018higher}).  Successive approximations
        can be thought of either as a robust alternative or complementary in the
        following sense: While projection methods are fast, they almost always require
        tuning, as output can be sensitive to both the choice of basis functions and
        the solver used for the associated nonlinear equations.  In such cases, the
        globally convergent successive approximations method can be employed to first
        compute the solution.  Projection methods can then be tuned until they
        reliably reproduce it.}

\begin{remark}\label{r:csra}
    In Appendix~\ref{s:ax} we use a local spectral radius result to show that,
    when Assumptions~\ref{a:p}--\ref{a:co} hold, we have
    \begin{equation}
        \label{eq:rvi}
        \lL_\Phi^p = \ln r(V)
    \end{equation}
    where $r(V)$ is the spectral radius of the operator $V$ introduced
    in~\eqref{eq:aop}, when regarded as
    a linear self-map on $L_p(\XX, \RR, \pi)$.  This result is central to the
    proof of Theorem~\ref{t:bk} and used in later
    computations.\footnote{The definition of the spectral radius is given in
        Appendix~\ref{s:ax}.  The spectral
        radius $r(V)$ also appears in the literature on stochastic discount factor
        decompositions discussed in the introduction, since, by the
        Krein--Rutman theorem, it equals
        the principal eigenvalue of the valuation operator $V$, which in turn
        determines the permanent growth component of the stochastic discount
    factor.}
\end{remark}

\begin{remark}\label{r:nis}
    The necessity component of Theorem~\ref{t:bk} has strong implications.
    To understand it from an asset pricing perspective, we interpret $g$ as
    defining a cash flow via~\eqref{eq:azdef} and view $h$ satisfying $h = Th$
    as a pricing function for this cash flow.  Theorem~\ref{t:bk} tells
    us that, when $\lL_\Phi^p \geq 0$, \emph{no} nontrivial cash flow with
    finite $p$-th moment can be priced
    under the discounting embedded in $V$.
    In other words, the only cash flow with finite price in $\hH_p$ is the cash flow
    that pays zero with probability one.
\end{remark}

It is also worth noting that, in view of (a) and (d) from Theorem~\ref{t:bk},
uniqueness of asset prices is never an issue under our assumptions.  Either
$\lL_\Phi^p < 0$ and one solution exists, or $\lL_\Phi^p \geq 0$ and no
solutions exist.

\section{Applications with Finite State Spaces}\label{s:aar}

We now turn to applications of Theorem~\ref{t:bk}, beginning with the 
classic study of \cite{mehra1985equity}.  Our
objective in treating this model is to clarify the assumptions and results in
Section~\ref{s:mr} in a simple environment, before moving on to more complex applications.  The
following proposition will aid our analysis.



\begin{proposition}
    \label{p:fc}
    If Assumptions~\ref{a:p}--\ref{a:i} hold and, in addition, the state space
    $\XX$ is finite, then Assumption~\ref{a:co} also holds, for every $p \geq
    1$, and $\lL_\Phi^p = \lL_\Phi$.  In particular (b)--(d) of
    Theorem~\ref{t:bk} all hold
    at every $p \geq 1$ if and only if $\lL_\Phi < 0$.
\end{proposition}

Proposition~\ref{p:fc} applies whenever the state evolves as a finite,
irreducible Markov chain---a common set up in quantitative
applications.\footnote{See, for example, \cite{backus_gregory_zin:1989},
    \cite{weil:1989}, \cite{kocherlakota:1990}, \cite{alvarez_jermann:2001},
\cite{cogley_sargent:2008}, \cite{collindufresne_johannes_lochstoer:2016}, or
\cite{martin_ross:2019}.} The proposition has two key implications.  One is
that, in the finite state setting, we can always work with the simple exponent
$\lL_\Phi^1 = \lL_\Phi$ from~\eqref{eq:prp}.  The reason the stability
exponent does not depend on $p$ is that moment conditions are irrelevant when
$\XX$ is finite, since all moments are finite for random variables supported
on finite sets.

The second implication is that the eventual compactness
condition in Assumption~\ref{a:co} is always satisfied when $\XX$ is finite.
Indeed, eventual compactness means that
there exists a time horizon $n$ such that
the $n$-period valuation operator $V^n$ is a compact linear operator (i.e., maps bounded sets of payoffs to
relatively compact sets).  This generalizes the idea that $V^n$ has finite rank
(i.e., maps into a finite dimensional range space).  When $\XX$ is finite, $V$
is just a matrix and the range space of $V^n$ is finite dimensional for all
$n$, so
eventual compactness certainly holds.

In \cite{mehra1985equity}, the price-dividend ratio $Q_t := P_t / D_t$
obeys~\eqref{eq:pd2}
and the stochastic discount factor is given by~\eqref{eq:ms}.
In other words,
\begin{equation}\label{eq:qt}
        Q_t
        = \EE_t \left[
                \beta \frac{u'(C_{t+1})}{u'(C_t)}
                \frac{D_{t+1}}{D_t}
                    \left(Q_{t+1} + 1\right)
                \right].
\end{equation}
Agents have CRRA utility
\begin{equation}
    \label{eq:cx}
    u(c) = \frac{c^{1 - \gamma}}{1 - \gamma}.
\end{equation}
In equilibrium, $C_{t+1}/C_t = D_{t+1}/D_t = X_{t+1}$.   The state space $\XX$ is
contained in $(0, \infty)$, so $X_t > 0$.  Equation~\eqref{eq:qt} becomes
\begin{equation}
    \label{eq:q_finite_Markov_chain}
    Q_t = \beta \, \EE_t \, X_{t+1}^{1-\gamma} \, [ Q_{t+1} + 1 ] ,
\end{equation}
which is a version of~\eqref{eq:fwl} with $\Phi_{t+1}
= \beta X_{t+1}^{1-\gamma}$ and $G_{t+1}=1$.  Since elements of $\XX$ are positive,
Assumption~\ref{a:p} holds.  \cite{mehra1985equity} assume that $\Pi(x,y) :=
\PP\{X_{t+1}=y \,|\, X_t = x\} > 0$ for all $x,y \in \XX$, so the
irreducibility condition in Assumption~\ref{a:i} holds.  In view of
Proposition~\ref{p:fc}, Assumption~\ref{a:co} is also valid
and $\lL_\Phi^p = \lL_\Phi$ for all $p \geq 1$.

With Assumptions~\ref{a:p}--\ref{a:co} all in force, Remark~\ref{r:csra}
applies, and we have $\lL_\Phi = \ln r(V)$, where $r(V)$ is the spectral
radius of the  valuation operator
\begin{equation*}
    V h(x)
    =  \EE_x \Phi_{t+1} \, h(X_{t+1})
    =  \beta \sum_{y \in \XX} y^{1-\gamma} h(y) \Pi(x, y).
\end{equation*}
In the present setting, the operator $V$ can be identified with the
matrix $V(x, y) = \beta y^{1-\gamma} \Pi(x, y)$, and $r(V)$ is just the
spectral radius of this matrix.

The intuition behind Remark~\ref{r:csra} is straightforward for this model.
Recall that $\lL_\Phi$ is the long-term decay rate in prices.
We now show that $\ln r(V)$ also has this interpretation.
The largest eigenvalue of the valuation matrix, equal to the spectral radius $r(V)$, dominates long-term
pricing. As the maturity $n$ increases, prices of cash flows $h$ behave as
$V^nh \sim r(V)^n e$, where $e$ is the eigenvector associated with the
largest eigenvalue. The long-term decay rate in prices
is, therefore,\footnote{Relatedly, \cite{backus_gregory_zin:1989} show that the yield
    on the long bond equals $-\ln r(V)$ for a valuation matrix
$V$ corresponding to $\Phi_{t+1} = \beta X_{t+1}^{-\gamma}$ in
\eqref{eq:q_finite_Markov_chain}.}
\begin{equation}
    \lim_{n\to\infty} \frac{1}{n} \ln V^n h = \ln r(V).
\end{equation}

Turning to Theorem~\ref{t:bk}, since Assumptions~\ref{a:p}--\ref{a:co} all
hold, statements (b)--(d) of Theorem~\ref{t:bk} are valid if and only if $\lL_\Phi
= \ln r(V) < 0$.  In \cite{mehra1985equity}, the transition
probabilities and state space are given by
\begin{equation*}
    \Pi =
    \begin{pmatrix}
        \psi & 1-\psi \\
        1-\psi & \psi
    \end{pmatrix}
    \quad \text{and} \quad
    \XX =
    \begin{pmatrix}
        1 + \mu + \delta \\
        1 + \mu - \delta
    \end{pmatrix}
\end{equation*}
for parameters $\psi, \mu, \delta$. The baseline
parameter values are $\psi = 0.43$, $\mu=0.018$, $\delta=0.036$, and
$\beta=0.99$.  The authors experiment with different values of $\gamma$
ranging from 1 to 10.  With $\gamma = 2.5$, evaluating the spectral radius of
$V$ gives $\ln r(V) = -0.0348$, so the equivalent conditions (b)--(d) in
Theorem~\ref{t:bk} hold.

The fact that $\ln r(V) < 0$ implies existence of a unique Markov solution in this
simple setting is well known: a Markov solution is a finite vector $h$ satisfying $h
= V(h + \1)$, where $\1$ is the unit vector, and the Neumann Series Theorem
yields a solution whenever $r(V) < 1$.  Nonetheless, Theorem~\ref{t:bk} is
useful even in this very simple case.  For example, it shows that
$\ln r(V) < 0$ is not just sufficient but also necessary for existence of a
finite price-dividend ratio.

Figure~\ref{f:mp_contours} helps illustrate the value of this exact
delineation.  The plot
shows contour lines for the value of $\lL_\Phi$ generated by the
Mehra--Prescott model and a set of neighboring parameterizations.   The
horizontal and vertical axes show grid points for the parameters $\gamma$ and
$\delta$ respectively.  The other parameters are held at the baseline
parameterization.  The contour line $\lL_\Phi = 0$ is emphasized.
Pairs $(\gamma, \delta)$ on the contour line $\lL_\Phi=0$, as well as those to
the far left (low
values of $\gamma$) and top right (high $\gamma$ and $\delta$) are where the
stability condition fails.  At other points, such as at the original
parameterization used by \cite{mehra1985equity}, we have $\lL_\Phi < 0$ and the
stability condition holds.

\begin{figure}
    \centering
    \scalebox{0.55}{\includegraphics[clip=true, trim=0mm 0mm 0mm 5mm]{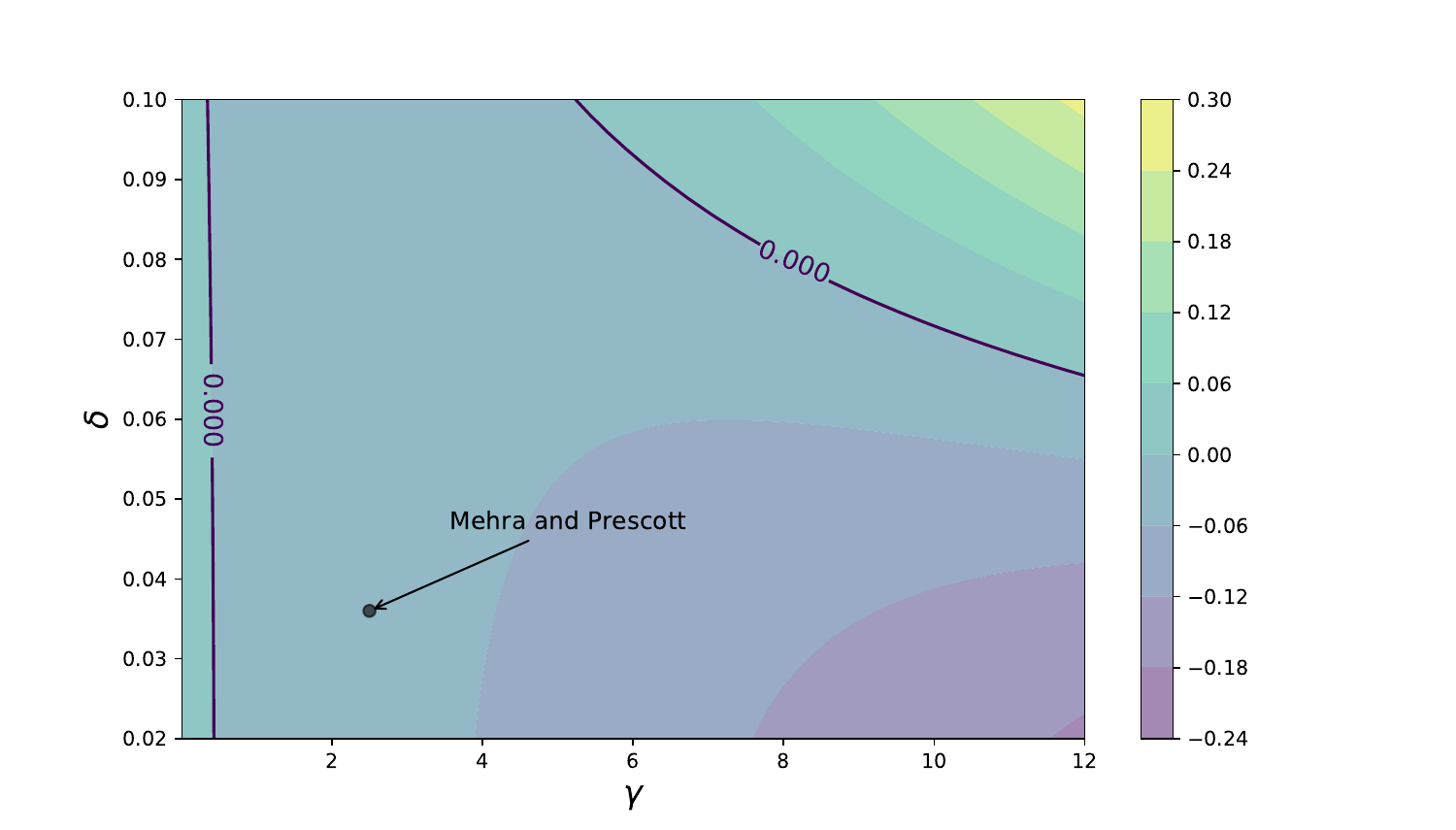}}
    \caption{\label{f:mp_contours} Contour plot of $\lL_\Phi$ for the Mehra-Prescott model}
\end{figure}

Recall that $\lL_\Phi$ can be interpreted as the (negative) rate of growth of
the price of long-horizon consumptions strips as maturity increases. The
nonmonotonicity of $\lL_\Phi$ is a consequence of the
changing strength of the substitution effect as the risk aversion parameter
$\gamma$ changes. An increase in the volatility parameter $\delta$ reduces the
certainty equivalent value of future consumption, acting as a negative income
effect. When $\gamma>1$, the low willingness to substitute consumption across
time dominates valuation. The value of future consumption in current
consumption units, and hence $\lL_\Phi$, increases as $\delta$ increases, and diverges to infinity in
the moment when the condition $\lL_\Phi<0$ fails. On the other hand, when
$\gamma<1$, the income effect dominates, and $\lL_\Phi$ becomes more negative
as $\delta$ increases. A solution does not exist for low values of $\gamma$
and $\delta$ because $\beta(1+\mu)>1$. In the case of logarithmic utility
($\gamma=1$), the income and substitution effect exactly offset each other,
and $\lL_\Phi = \log \beta <0$.

\section{Applications with Log Linear Growth}

\label{ss:crra}

Next we study two asset pricing applications where dynamics are linear and shocks
are Gaussian.  The main objective of this section is to show how, in some
settings, the value $\lL_\Phi^p$ can be calculated analytically, as well as to
derive intuition on the factors that determine $\lL_\Phi^p$. We will also use
the analytical solutions as a benchmark for testing numerical
calculations (see Appendix~\ref{s:ctv}).

\subsection{CRRA Preferences and Log Linear Growth}

As a first step, we replace the finite state process in \cite{mehra1985equity}
with a linear Gaussian process.   In particular, dividend and consumption
growth are now assumed to obey the constant volatility specification from section~I.A of
\cite{bansal2004risks}, which is
\begin{subequations}
\begin{align}
    \ln \left( D_{t+1}/D_{t}\right)
        & = \mu_d + \varphi X_t + \sigma_d \, \xi_{t+1}
    \label{al:sar000}
    \\
    \ln \left( C_{t+1}/C_{t}\right)
        & = \mu_c + X_t + \sigma_c \, \epsilon_{t+1}
    \label{al:sar00}
        \\
    X_{t+1}
        &=\rho X_t + \sigma \, \eta _{t+1}
    \label{al:sar0}
\end{align}
\end{subequations}
Here $-1 < \rho < 1$ and $\left\{ (\xi_t, \epsilon_t,  \eta_t) \right\} $ is
{\sc iid} and standard normal in $\RR^3$.  This is a natural extension of the
original model to a continuous state setting. 

The model is otherwise unchanged.
The SDF has the standard time separable form \eqref{eq:ms} and agents have
CRRA utility.  
We solve for the price dividend ratio using \eqref{eq:pd2}, which
means that, when connecting to the forward looking model~\eqref{eq:fwl}, we
take $G_t = 1$ and
\begin{equation}
    \label{eq:phimp}
    \Phi_{t+1}
    = M_{t+1} \frac{D_{t+1}}{D_t}
        = \beta \, \exp \left\{
            (\mu_d + \varphi X_t + \sigma_d \, \xi_{t+1})
            -\gamma
            (\mu_c + X_t + \sigma_c \, \epsilon_{t+1})
        \right\}.
\end{equation}
(See the discussion following \eqref{eq:fwl}.)

\begin{proposition}\label{p:lpc}
    For $\{\Phi_t\}$ in~\eqref{eq:phimp} and $p \geq 1$, we have
    \begin{equation}
        \label{eq:lpc}
        \lL_\Phi^p
        = \ln \beta
            +
            \mu_d - \gamma \mu_c
            + \frac{\sigma^2}{2} \frac{(\varphi-\gamma)^2}{(1-\rho)^2}
            + \frac{\sigma_d^2 + (\gamma \sigma_c)^2}{2}.
    \end{equation}
\end{proposition}

The proof of Proposition~\ref{p:lpc} starts from the definition
in~\eqref{eq:prpp} and steps through straightforward calculations.  The proof
is in Appendix~\ref{ss:lpc}.

The value of $\lL_\Phi^p$ represents the long-run growth rate of the
discounted dividend $\Phi_t$. In expression (\ref{eq:lpc}), the term $\mu_d
+\sigma_d^2/2 + [\varphi\sigma/(1-\rho)]^2/2$ corresponds to the long-run
dividend growth rate, $\ln\beta-\gamma\mu_c + (\gamma\sigma_c)^2/2 +
[\gamma\sigma/(1-\rho)]^2/2$ to the (negative of) the long-run discount rate,
and $\varphi\gamma\sigma^2/(1-\rho)^2$ is the long-run covariance between the
two.\footnote{Notice that $\lL_\Phi^p$ in \eqref{eq:lpc} does not depend on $p$.  In
particular, we have $\lL_\Phi^p = \lL_\Phi^1 := \lL_\Phi$ for all $p$.
This matches the finding that $\lL_\Phi^p = \lL_\Phi$ for all $p$ in the
finite dimensional case, as shown in Proposition~\ref{p:fc}.  In other words, this
constant volatility model is simple enough to retain key features of
the finite dimensional setting.}

When do the conditions of Theorem~\ref{t:bk} hold?
Since $G_t = 1$ and $\Phi_t$ is given by \eqref{eq:phimp},
Assumption~\ref{a:p} is clearly valid.  The state process \eqref{al:sar0} is
irreducible, so Assumption~\ref{a:i} holds.
For the moment condition $\EE (\Phi_t G_t)^p < \infty$ in
Assumption~\ref{a:co}, a finite $p$-th moment
in~\eqref{eq:phimp} is required.  This holds for all $p \geq 1$ in the current setting
because the stationary distribution of $X_t$ is
Gaussian.

The only remaining issue is the eventual compactness of $V$ in Assumption~\ref{a:co}.
In view of \eqref{eq:abel}, the valuation operator
$V$ has the form
\begin{equation*}
    Vh(x) = \beta \, \exp \left\{a x + b \right\} \int h(y) q(x, y) \diff y
\end{equation*}
for suitably chosen constants $a$ and $b$,
where $q$ is the Gaussian transition density associated with \eqref{al:sar0}.
From this expression it can be verified that Assumption~\ref{a:co} holds at
$p=2$ via Proposition~\ref{p:eco} in the appendix.  This follows from the
smoothing property of conditional expectations, which implies that the range
of the operator $V$ is not too irregular.

We conclude that, since Assumptions~\ref{a:p}--\ref{a:co} are all satisfied,
the conclusions of Theorem~\ref{t:bk} hold if and only if
the right hand side of \eqref{eq:lpc} is negative.

\subsection{Habit Persistence}

\label{ss:ahpm}

There is a large literature on asset prices in the presence of consumption
externalities and habit formation (see, e.g., \cite{abel1990asset} and
\cite{campbell1999force}). In this section we treat a relatively simple habit
formation model and show how the stability exponent can be calculated
analytically.  In the process, we illustrate the value of Theorem~\ref{t:bk}
by substantially improving on the existence and uniqueness results in
\cite{calin2005solving}.

In the ``external'' habit formation setting of \cite{abel1990asset} and \cite{calin2005solving},
the growth adjusted SDF takes the form
\begin{equation}
    \label{eq:abel}
    M_{t+1} \frac{D_{t+1}}{D_t}
        = k_0 \exp((1 - \gamma)(\rho - \alpha) X_t)
\end{equation}
where $k_0 := \beta \exp(b(1 - \gamma) + \sigma^2(\gamma - 1)^2/2)$ and
$\alpha$ is a preference parameter.    The
state sequence $\{X_t\}$ obeys
\begin{equation}
    \label{eq:sar0}
    X_{t+1} = \rho X_t + b + \sigma \, \eta _{t+1}
    \quad \text{with} \quad
    -1 < \rho < 1
    \quad \text{and} \quad \{\eta_t\} \iidsim N(0, 1).
\end{equation}
The parameter $b$ is equal to $x_0 + \sigma^2 (1-\gamma)$ where $x_0$
represents mean constant growth rate of the dividend of the asset.

The price-dividend ratio associated with this stochastic discount factor
satisfies the forward recursion \eqref{eq:pd2} and,
by Theorem~\ref{t:bk}, there exists a unique price-dividend with finite second
moment (we set $p=2$ in Theorem~\ref{t:bk}) if $\lL_\Phi^2 < 0$ and
Assumptions~\ref{a:p}--\ref{a:co} are satisfied.
Assumptions~\ref{a:p}--\ref{a:co} can be verified when $p=2$ in almost
identical manner to the corresponding discussion in Section~\ref{ss:crra}.
Hence, by Theorem~\ref{t:bk}, a unique equilibrium price-dividend ratio with
finite second moment exists if and only if
$\lL_\Phi^2 < 0$.



An analytical expression for $\lL_\Phi^2$ can be obtained using similar
techniques to those employed in Section~\ref{ss:crra}.  Stepping through
the algebra shows that
\begin{equation}
    \label{eq:sea}
    \lL_\Phi^2
    = \ln k_0
        +
            (1 - \gamma)(\rho - \alpha)
            \frac{b}{1-\rho} + \frac{(1 - \gamma)^2(\rho - \alpha)^2}{2} \frac{\sigma^2}{(1-\rho)^2}
            .
\end{equation}
A unique equilibrium price-dividend ratio exists in $\hH_2$ if and only this
term is negative.  The intuition behind the expression \eqref{eq:sea} is
analogous to \eqref{eq:lpc} in Section~\ref{ss:crra}.

To give some basis for comparison, let us contrast the condition $\lL_\Phi^2 <
0$ with the sufficient condition for existence and uniqueness of an
equilibrium price-dividend ratio found in Proposition~1 of
\cite{calin2005solving}, which implies a one step contraction.  Their test is of
the form $\tau < 1$, where $\tau$ depends on the parameters of the model (see
Equation (7) of \cite{calin2005solving} for details).  Since the condition
$\lL_\Phi^2 < 0$ requires only eventual contraction, rather than one step
contraction, we can expect it to be
significantly weaker than the condition of \cite{calin2005solving}.

Figure~\ref{f:abel1} supports this conjecture.
The left sub-figure shows $\ln \tau$ at a range of parameterizations.  The
right sub-figure shows $\lL_\Phi^2$ at the same parameters, evaluated using
\eqref{eq:sea}.  The horizontal and vertical axes show grid points for the
parameters $\beta$ and $\sigma$ respectively.  For both sub-figures, $(\beta,
\sigma)$ pairs with
test values strictly less than zero (points to the south west of the 0.0
contour line) are where the respective condition holds.  Points to the north
west of this contour line are where it fails.\footnote{The parameters held
    fixed in Figure~\ref{f:abel1} are $\rho=-0.14$, $\gamma=2.5$, $x_0 = 0.05$
and $\alpha=1$.}

Inspection of the figure shows that the sufficient condition in
\cite{calin2005solving} fails for some empirically relevant parameterizations
that have unique stationary Markov equilibria.  Note also that, because
$\lL_\Phi^2 < 0$ is both necessary and sufficient in our setting, the 0.0
contour line in the right sub-figure is an exact delineation between stable
and unstable parameterizations.

\begin{figure}
    \centering
    \scalebox{0.55}{\includegraphics[clip=true, trim=0mm 0mm 0mm 5mm]{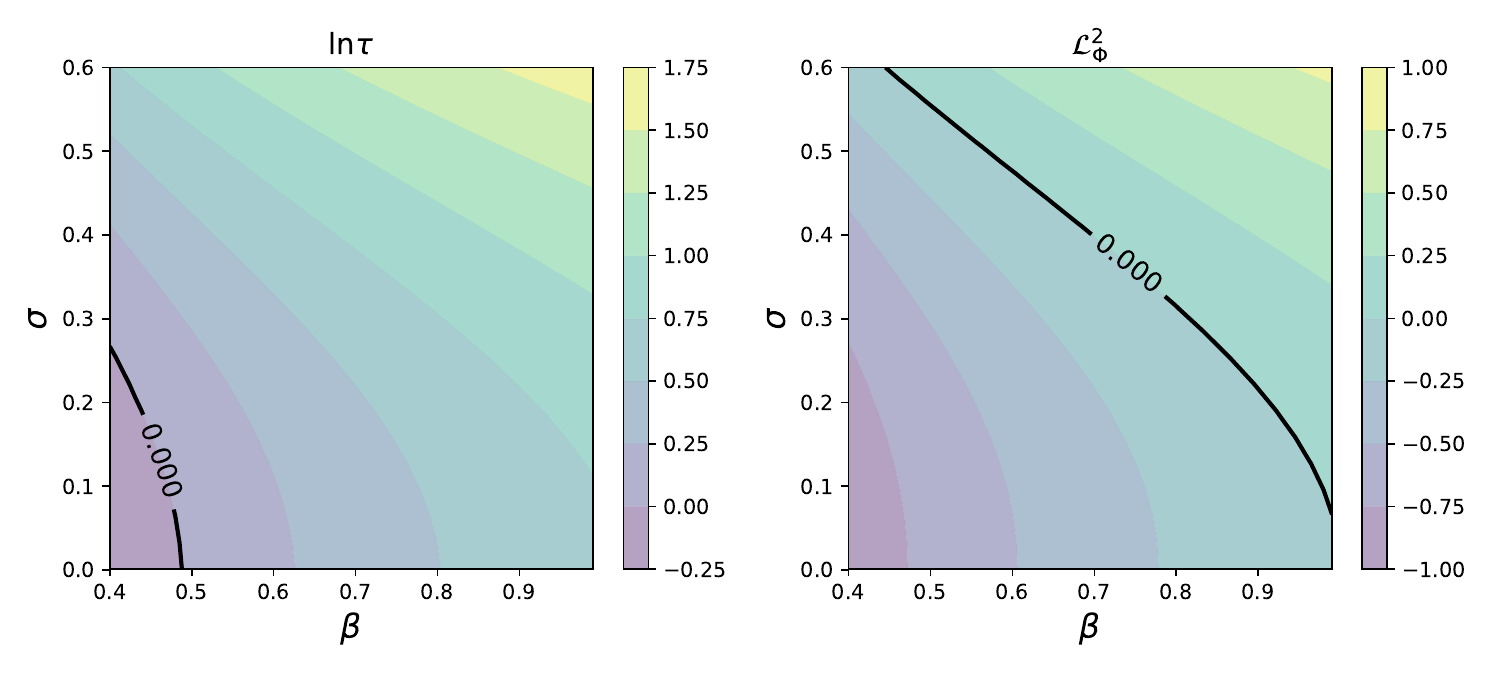}}
    \caption{\label{f:abel1} Alternative tests of stability for the habit
    formation model}
\end{figure}



\section{Applications with Long-Run Risk}\label{s:aii}

In this section we continue to apply Theorem~\ref{t:bk}, but now in the
presence of more sophisticated models, with time-varying risk.  Our aim is to
show how Theorem~\ref{t:bk} can clearly delineate between well-defined models
and models with no solution, even in settings with features such as recursive
preferences and nonlinear dynamics (e.g., stochastic volatility).
We focus on the class of long-run risk models first developed by
\cite{bansal2004risks}, which have generated insights in a range of quantitative
applications.

\subsection{Long-Run Risk With Stochastic Volatility}

\label{ss:lri}

Next we turn to an asset pricing model with Epstein--Zin utility and
stochastic volatility in cash flow and consumption estimated by
\cite{bansal2004risks}.  Preferences are represented by the continuation value
recursion
\begin{equation}
	\label{eq:agg}
	V_t = \left[
            (1 - \beta) C_t^{1-1/\psi}
            + \beta \left\{ \rR_t \left(V_{t+1}
            \right) \right\}^{1-1/\psi}
          \right]^{1/(1-1/\psi)},
\end{equation}
where $\{C_t\}$ is the consumption path and $\rR_t$ is the certainty
equivalent operator
\begin{equation}
	\label{eq:ce}
	\rR_t(Y)
    := ( \EE_t  Y^{1-\gamma} )^{1/(1-\gamma)}.
\end{equation}
The parameter $\beta \in (0, 1)$ is a time discount factor, $\gamma$
governs risk aversion and $\psi$ is the intertemporal elasticity of
substitution.  Dividends and consumption grow according to
\begin{subequations}
    \label{eq:byd}
    \begin{align}
        \ln (C_{t+1} /  C_t)
        & = \mu_c + z_t + \sigma_t \, \eta_{c, t+1},
        \\
        \ln (D_{t+1} /  D_t)
        & = \mu_d + \alpha z_t + \varphi_d \, \sigma_t \, \eta_{d, t+1},
        \\
        z_{t+1}
        & = \rho z_t + \varphi_z \, \sigma_t \, \eta_{z, t+1},
        \\
        \sigma^2_{t+1}
        & = \max
                \left\{
                    v \, \sigma_t^2 + d + \varphi_\sigma \, \eta_{\sigma, t+1},\, 0
                \right\}.
    \end{align}
\end{subequations}
Here $\{\eta_{i, t}\}$ are {\sc iid} and standard normal for $i \in \{d, c, z,
\sigma\}$.  The state $X_t$ can be represented as $X_t =
(z_t, \sigma_t)$.  The (growth adjusted) SDF process associated with this model is
\begin{equation}
    \label{eq:phiez}
    \Phi_{t+1}
    := M_{t+1} \frac{D_{t+1}}{D_t}
    = \beta^\theta
        \frac{D_{t+1}}{D_t}
        \left(
            \frac{C_{t+1}}{C_t}
        \right)^{-\gamma}
        \left(
            \frac{W_{t+1}}{W_t - 1}
        \right)^{\theta - 1},
\end{equation}
where $W_t$ is the aggregate wealth-consumption ratio
and $\theta := (1-\gamma)/(1-1/\psi)$.\footnote{For a derivation see, for example,
\cite{bansal2004risks}, p.~1503.}
To obtain the aggregate wealth-consumption ratio $\{W_t\}$ we exploit the fact that
$W_t = w(X_t)$ where $w$ solves the Euler equation
\begin{equation*}
    \beta^\theta \,
    \EE_t
    \left[
        \left(
            \frac{C_{t+1}}{C_t}
        \right)^{1-\gamma}
        \left(
            \frac{w(X_{t+1})}{w(X_t) - 1}
        \right)^\theta
    \right]
    = 1.
\end{equation*}
Rearranging and using the expression for consumption growth given above, this
equality can be expressed as
\begin{equation*}
    \label{eq:sfw}
    w(z, \sigma) = 1 +  [ Kw^\theta (z, \sigma) ]^{1/\theta},
\end{equation*}
where $K$ is the operator
\begin{equation}
    \label{eq:defk}
    K g(z, \sigma)
    = \beta^\theta
        \exp
        \left\{
        (1-\gamma) (\mu_c + z)
        + \frac{(1-\gamma)^2\sigma^2}{2}
        \right\}
        \Pi g(z, \sigma)
\end{equation}
In this expression, $\Pi g(z, \sigma)$ is the expectation of $g(z_{t+1},
\sigma_{t+1})$ given the state's law of motion, conditional on $(z_t, \sigma_t) =
(z, \sigma)$.

The existence of a unique solution $w = w^*$ to \eqref{eq:sfw} in $\hH_1$
under the parameterization used in
\cite{bansal2004risks} is established in \cite{borovivcka2020necessary}
when
the innovation terms $\{\eta_{i, t}\}$ are truncated, so that the state space
is a compact subset of $\RR^2$.  In what follows, we compute $w^*$ using the
iterative method described in \cite{borovivcka2020necessary}
and recover $W_t$ as
$w^*(X_t)$ for each $t$.

As discussed in detail in Appendix~\ref{s:ctv}, to approximate the stability
exponent $\lL_\Phi$, we
can use Monte Carlo, generating independent paths for the SDF process
$\{\Phi_t\}$ and averaging over them to estimate the expectation on the right
hand side of \eqref{eq:prp}.  In computing the product $\prod_{t=1}^n \Phi_t$
we used \eqref{eq:byd} and \eqref{eq:phiez} to express it as
\begin{multline}
    \prod_{t=1}^n \Phi_t
    =
    (\beta^\theta \exp(\mu_d - \gamma \mu_c))^n \label{eq:Phi_BY}
    \\
    \times
    \exp
    \left(
        (\alpha - \gamma) \sum_{t=1}^n z_t
        -
        \gamma \sum_{t=1}^n \sigma_t \eta_{c,t+1}
        +
        \varphi_d \sum_{t=1}^n \sigma_t \eta_{d,t+1}
        +
        (\theta - 1) \sum_{t=1}^n \hat w_t
    \right),
\end{multline}
where $\hat w_{t+1} = \ln[ W_{t+1}/(W_t - 1)]$.

At the parameter values using in \cite{bansal2004risks} and based on the Monte
Carlo method discussed above, we estimate that
$\lL_\Phi = -0.00388$ implying the existence of a unique equilibrium price-dividend ratio
  function in $\hH_1$.\footnote{The reported value is the mean of 1,000
      draws of $\lL_\Phi(n, m)$, where
      the latter is defined in \eqref{eq:mc} of Appendix~\ref{s:ctv}.
      For each draw, $n$ and $m$ in
      in this calculation were set to 1,000 and 10,000 respectively.  The
      standard error for the mean was approximately 0.0001.
      Following \cite{bansal2004risks}, the
      parameters used were $\gamma=10.0$, $\beta=0.998$, $\psi=1.5$ $\mu_c=0.0015$,
      $\rho=0.979$, $\varphi_z = 0.044$, $v=0.987$, $d=$7.9092e-7,
      $\varphi_\sigma=$2.3e-6.  $\mu_d=0.0015$, $\alpha=3.0$ and $\varphi_d=4.5$.
  See table IV on page~1489.} While this value is close to zero, we find that
  significant shifts in parameters are
  required to cross the contour $\lL_\Phi = 0$.

  For example,
  Figure~\ref{f:alpha_mud_by} shows $\lL_\Phi$ calculated at a
  range of parameter values in the neighborhood of the \cite{bansal2004risks}
  specification via a contour map.  The parameter $\alpha$ is varied on the
  horizontal axis, while $\mu_d$ is on the vertical axis.  Other parameters
  are held fixed at the \cite{bansal2004risks} values.  The black contour line
  shows the boundary between stability and instability.  Not surprisingly, the
  test value increases with the cash flow growth rate $\mu_d$. In this region
  of the parameter space, it also declines with $\alpha$, because an increase
  in $\alpha$ with $\gamma > \alpha$ reduces the covariance between cash flow
  growth and discounting captured by the term $(\alpha - \gamma) \sum_{t=1}^n
  z_t$ in (\ref{eq:Phi_BY}).

\begin{figure}
    \centering
    \scalebox{0.55}{\includegraphics[clip=true, trim=0mm 0mm 0mm 0mm]{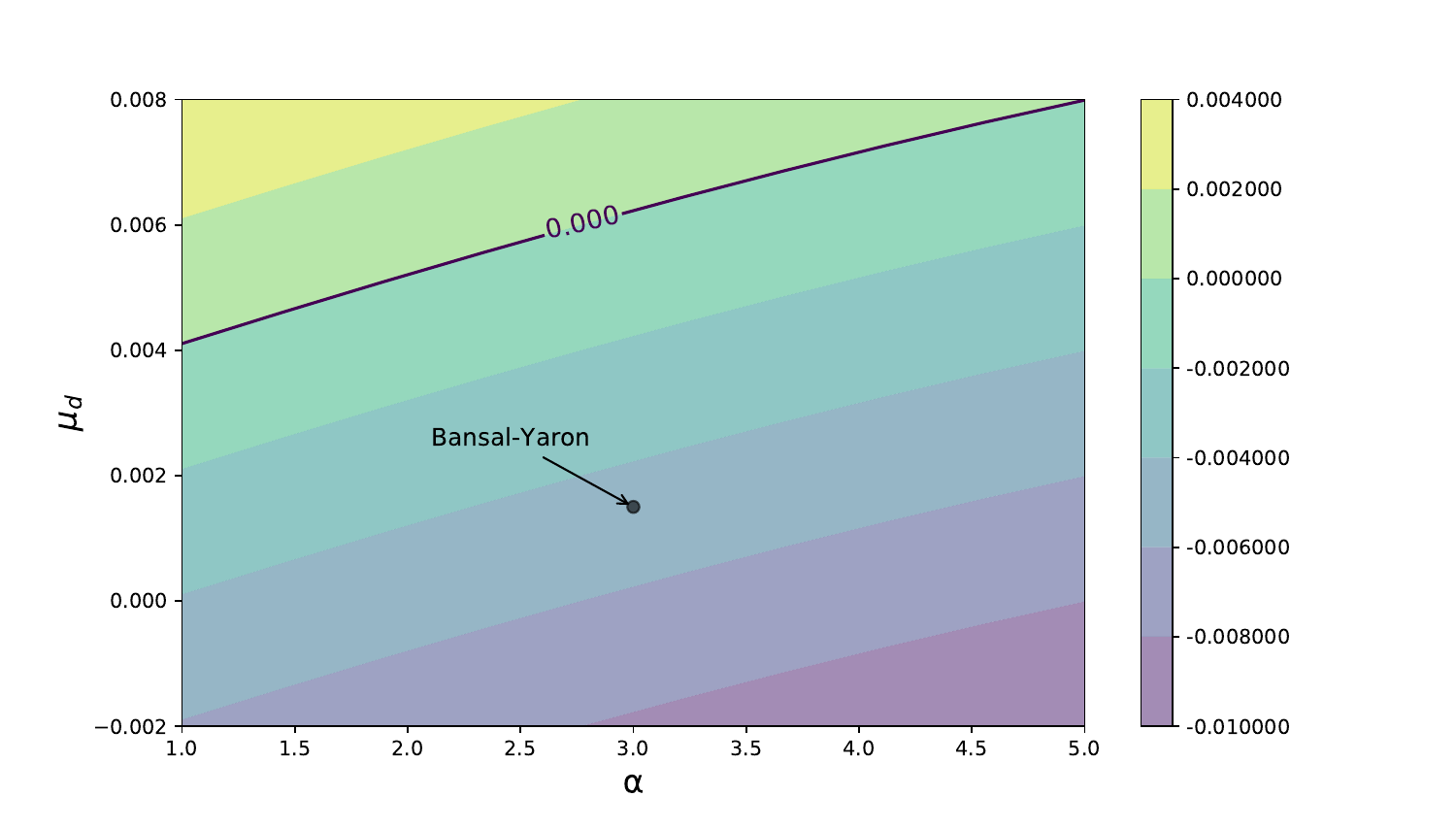}}
    \vspace{0.1em}
    \caption{\label{f:alpha_mud_by} The exponent $\lL_\Phi$ for the Bansal--Yaron model}
\end{figure}

\subsection{Long-Run Risk Part II}

\label{ss:lrii}

Now we repeat the analysis in Section~\ref{ss:lri} but using instead
the dynamics for consumption and dividends in
\cite{schorfheide2018identifying}, which are given by
\begin{align*}
    \ln (C_{t+1} /  C_t)
    & = \mu_c + z_t + \sigma_{c, t} \, \eta_{c, t+1},
    \\
    \ln (D_{t+1} /  D_t)
    & = \mu_d + \alpha z_t + \delta \sigma_{c, t} \, \eta_{c, t+1} + \sigma_{d, t} \, \eta_{d, t+1},
    \\
    z_{t+1}
    & = \rho \, z_t + (1 - \rho^2)^{1/2} \, \sigma_{z, t} \, \upsilon_{t+1},
    \\
    \sigma_{i, t}
    & = \phi_i \, \bar{\sigma} \exp(h_{i, t}),
    \\
    h_{i, t+1}
    & = \rho_{h_i} h_i + \sigma_{h_i} \xi_{i, t+1},
    \quad i \in \{z, c, d\}.
\end{align*}
The innovation vectors $\eta_t = (\eta_{c, t}, \eta_{d, t})$ and $\xi_t := (\upsilon_t, \xi_{z, t}, \xi_{c, t}, \xi_{d, t})$
are {\sc iid} over time, mutually independent and standard normal in $\RR^2$
and $\RR^4$ respectively.  The state can be represented as
$X_t := (z_t, h_{z, t}, h_{c, t}, h_{d, t})$.  Otherwise
the analysis and methodology radius is similar to Section~\ref{ss:lri}.  The
product of the growth adjusted stochastic discount factors over $n$ period from
$t=1$ is
\begin{multline*}
    \prod_{t=1}^n \Phi_t
    =
    (\beta^\theta \exp(\mu_d - \gamma \mu_c))^n
    \\
    \exp
    \left(
        (\alpha - \gamma) \sum_{t=1}^n z_t
        +
        (\delta - \gamma) \sum_{t=1}^n \sigma_{c,t} \eta_{c,t+1}
        +
        \sum_{t=1}^n \sigma_{d,t} \eta_{d,t+1}
        +
        (\theta - 1) \sum_{t=1}^n \hat w_t
    \right)
\end{multline*}
As in Section~\ref{ss:lri}, we generate this product many times and then
average to obtain an approximation of $\lL_\Phi$.
At the parameterization used in \cite{schorfheide2018identifying},
this evaluates to $-0.001$, indicating the existence of a unique
equilibrium price dividend ratio.\footnote{We used the posterior mean values from
    \cite{schorfheide2018identifying}, setting
                 $\beta=0.999$, $\gamma=8.89$, $\psi=1.97$,
                 $\mu_c=0.0016$,
                 $\rho=0.987$,
                 $\phi_z=0.215$,
                 $\bar \sigma=0.0032$,
                 $\phi_c=1.0$,
                 $\rho_{hz}=0.992$,
                 $\sigma_{hz}=\sqrt{0.0039}$,
                 $\rho_{hc}=0.991$,
                 $\sigma_{hc}=\sqrt{0.0096}$,
                 $\mu_d=0.001$,
                 $\alpha=3.65$,
                 $\delta=1.47$,
                 $\phi_d=4.54$,
                 $\rho_{hd}=0.969$, and
             $\sigma_{hd}=\sqrt{0.0447}$.
         We set $n=1,000$ and $m=10,000$, and then drew 1,000 observations of
     the statistic $\lL_\Phi(n, m)$, as defined in \eqref{eq:mc} of
 Appendix~\ref{s:ctv}.  The mean of these 1,000 draws was $-0.00103$, with
 standard error $0.000008$.}

Figure~\ref{f:mud_phid_ssy} shows the stability exponent $\lL_\Phi$ calculated
at a range of parameter values in the neighborhood of the
\cite{schorfheide2018identifying} specification.  The parameter $\phi_d$ is
varied on the horizontal axis, while $\mu_d$ is on the vertical axis.  Other
parameters are held fixed at the \cite{schorfheide2018identifying} values.
The interpretation is analogous to that of Figure~\ref{f:alpha_mud_by} from
Section~\ref{ss:lri}, as is the method of computation, with the dark contour
line shows the exact boundary between stability and instability. Increases in
both $\mu_d$ and $\phi_d$ increase the long-run growth rate of the level of
the discounted cash flow, and hence increase $\lL_\Phi$.  As with
Figure~\ref{f:alpha_mud_by}, significant deviations in estimated parameter
values are required to change the sign of $\lL_\Phi$.

\begin{figure}
    \centering
    \scalebox{0.55}{\includegraphics[clip=true, trim=0mm 0mm 0mm 0mm]{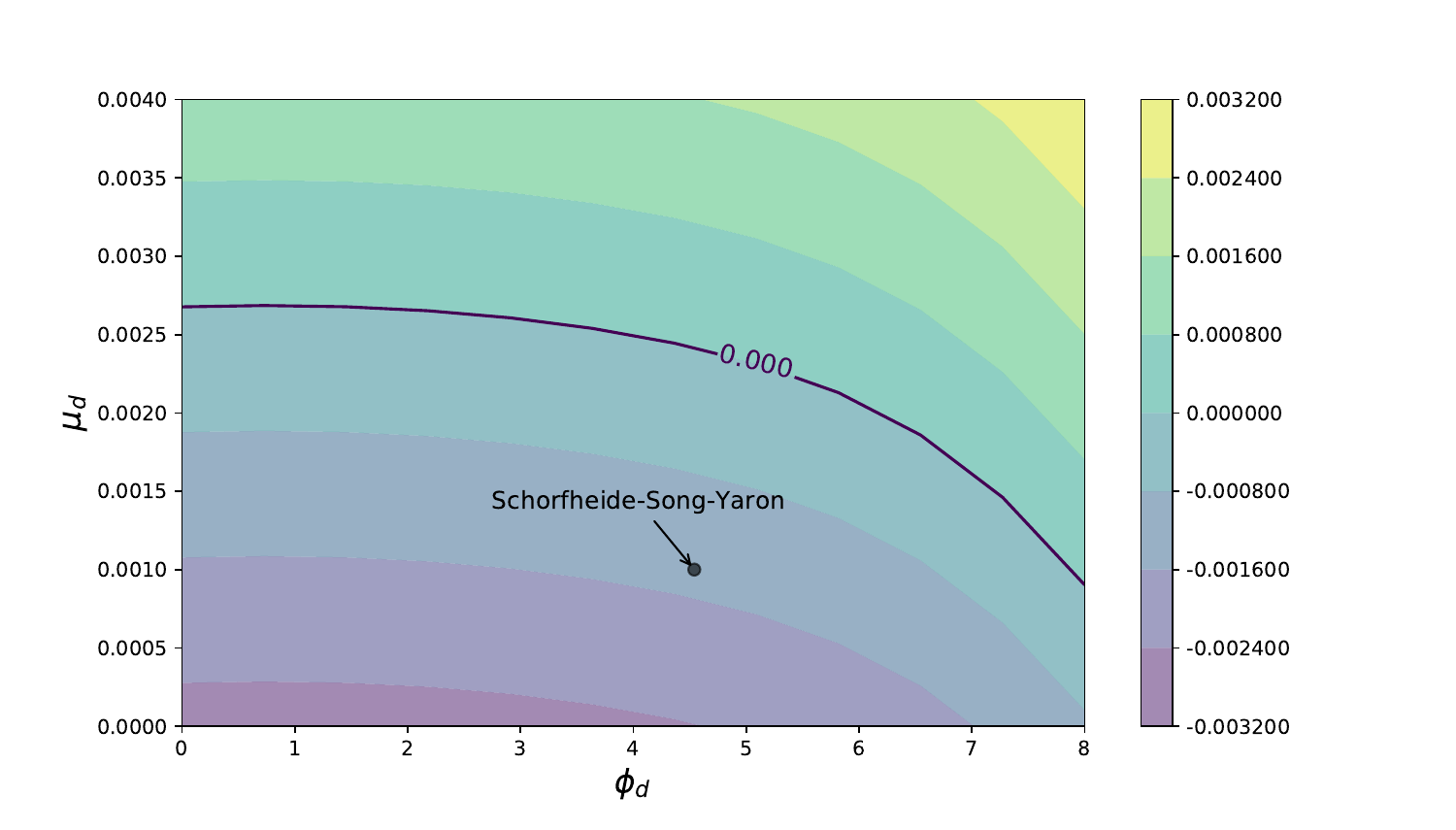}}
    \vspace{0.1em}
    \caption{\label{f:mud_phid_ssy} The exponent $\lL_\Phi$ for the Schorfheide--Song--Yaron model}
\end{figure}

\section{Conclusion}

\label{s:c}

In this paper we developed a practical test for existence and uniqueness of
equilibrium asset prices in infinite horizon arbitrage free settings.  By
seeking restrictions that ensure contraction occurs ``on average,
eventually,'' we obtained a test that is necessary as well as sufficient, and
hence yields an exact delineation between stable and unstable models.
Computational techniques are provided to ensure that the test can be
implemented in realistic quantitative applications.

It is natural to ask whether or not our results extend to a continuous time
setting.  We have provided an online appendix which shows that, at least in
simple cases, the answer is affirmative.  However, we treat no substantial
applications in that note and only briefly touch on interesting connections
between infinitesimal descriptions and stability results.  We hope that at
least some readers will pursue this research further.

Although we focused on consumption-based asset pricing models,
the theoretical results apply in the same way to other no-arbitrage
settings where asset prices can be represented using recursion (\ref{eq:pd})
with a positive marginal rate of substitution. Embedding this analysis into
frameworks with endogenously determined consumption is left to future
research.

\appendix

\section{Models with Stationary Dividends}\label{s:aiii}

In this section we discuss asset pricing with stationary dividends, rather than
stationary dividend growth (which is a more standard assumption in
quantitative analysis).  This topic is mainly of theoretical and historical
interest.  One aim is to extend the classical result on
existence and uniqueness of equilibrium asset prices obtained using
contraction mapping arguments in \cite{lucas1978asset}.

In that study, the price process obeys
\eqref{eq:pd}, where $P_t$ is the price of a claim to the aggregate endowment
stream $\{D_t\}$, and the stochastic discount factor is as given in
\eqref{eq:ms}.  In equilibrium, $C_t$ is equal to an endowment $D_t$, which is
a positive and continuous function of a stationary Markov
process $\{X_t\}$.  Following \cite{lucas1978asset} we  divide the fundamental asset
pricing equation \eqref{eq:pd} through by $u'(C_t)$ and set $D_t = C_t$
for all $t$, obtaining
\begin{equation}
    \label{eq:yl}
    Y_t
    = \beta \, \EE_t [ Y_{t+1} + u'(C_{t+1}) C_{t+1} ]
\end{equation}
where $Y_t := P_t \, u'(C_t)$.  This is a version of \eqref{eq:fwl} with
$\Phi_t = \beta$ and $G_t =u'(C_t) C_t$.

\cite{lucas1978asset} requires that the utility function $u$ is bounded, in
order to employ a contraction mapping theorem in a space of bounded functions.
This assumption
is violated in almost all quantitative applications.
To address this issue, \cite{brogueira2017existence} take
    $\XX = \RR$, set utility to be CRRA as in \eqref{eq:cx}, and suppose that
    $C_t = D_t = c(X_t)$ where $c(x) := a \exp(x)$ for some $a > 0$.
    For the state process $\{X_t\}$ they suppose that $\{X_t\}$ has a
    Gaussian density kernel $q(x, y)$ of the form $q(x, y) = N(\rho x,
    \sigma)$ for some $\sigma > 0$ and $|\rho | < 1$.
    The constant discount parameter $\beta$ is assumed to satisfy $\ln \beta < - (1-\gamma)^2 \sigma^2/2$.

The conditions of Theorem~\ref{t:bk} hold under these conditions when $p = 2$.
Assumptions~\ref{a:p} and \ref{a:i} are obviously true.
The moment condition in Assumption~\ref{a:co} holds when $p=2$ because
\begin{equation*}
    (\Phi_t G_t )^2
    = [ \beta u'(c(X_t)) c(X_t) ]^2
    = \beta^2 \exp(2 (1-\gamma) X_t).
\end{equation*}
The expectation of this term is finite because $X_t$ is Gaussian.  In
addition, $Vh(x) = \beta \int h(y) q(x, y) \diff y$ is an eventually compact
linear operator on $L_2(\XX, \RR, \pi)$,  as shown in Proposition~\ref{p:eco}.
Finally, since $\Phi_t = \beta$ for all $t$, we have $\lL_\Phi^2 = \ln \beta <
0$.  Hence, the conclusions of Theorem~\ref{t:bk} all hold.  In particular, a
unique equilibrium price process with finite second moment exists.
Notice that we did not require the stronger restriction $\ln \beta < - (1-\gamma)^2
\sigma^2/2$ from \cite{brogueira2017existence}.

\section{Computing the Stability Exponent}\label{s:ctv}

The stability exponent $\lL_\Phi^p$ plays a key role our results.  In some
cases it can be calculated analytically, as in \eqref{eq:lpc} or
\eqref{eq:sea}.   In others it needs to be computed.  We begin with a
discussion of the first case.

\subsection{Analytical Results}\label{ss:lpc}

In this section we provide the proof of Proposition~\ref{p:lpc}, which
illustrates how $\lL_\Phi^p$ can be calculated analytically in a constant
volatility setting.

\begin{proof}[Proof of Proposition~\ref{p:lpc}]
    From~\eqref{eq:phimp}, we have
    \begin{equation*}
        \prod_{i=1}^n \Phi_i
        = \beta^n \exp
            \left\{
                n (\mu_d -\gamma \mu_c)
            + (\varphi-\gamma) \sum_{i=1}^n X_i
            + \sigma_d \sum_{i=1}^n \xi_i
            - \gamma \sigma_c \sum_{i=1}^n \epsilon_i
            \right\}.
    \end{equation*}
    Using \eqref{al:sar0}, we then have
    \begin{equation}
        \label{eq:cbu}
        \left( \EE_x \prod_{i=1}^n \Phi_i \right)^p
        = \beta^{np} \exp ( p a_n x + p b_n ),
    \end{equation}
    where $a_n := (\varphi- \gamma) \rho (1-\rho^n)/(1-\rho) $ and
    \begin{equation*}
        b_n := n (\mu_d - \gamma \mu_c)
             +  \frac{(\varphi-\gamma)^2 s^2_n + n \sigma_d^2 + n (\gamma \sigma_c)^2}{2} .
    \end{equation*}
    Here $s_n^2$ is the variance of $\sum_{i=1}^n X_i$.  The next step in
    calculating $\lL_\Phi^p$ is to take the unconditional expectation of
    \eqref{eq:cbu}, which amounts to integrating with respect to the stationary
    distribution $\pi = N(0, \sigma^2 / (1- \rho^2))$.  This yields
    \begin{equation*}
        \EE \left( \EE_x \prod_{i=1}^n \Phi_i \right)^p
        = \beta^{np}
        \exp \left( \frac{(p a_n \sigma)^2}{2(1 - \rho^2)}  + p b_n \right),
    \end{equation*}
    and hence
    \begin{equation}
        \label{eq:cbu2}
        \lL_\Phi^p
        = \lim_{n \to \infty}
        \left\{
            \ln \beta + \frac{p}{n} \frac{(a_n \sigma)^2}{2(1 - \rho^2)}  +
                \frac{b_n}{n}
        \right\}
        =
        \ln \beta + \lim_{n \to \infty} \frac{b_n}{n} ,
    \end{equation}
    where the second equality uses the fact that $a_n$ converges to a finite constant.
    Some algebra yields
    \begin{equation}
        \label{eq:av}
        \frac{s^2_n}{n}
        = \frac{\sigma^2}{1-\rho^2}
        \left\{
            1 +
             \frac{2(n-1)}{n} \frac{\rho}{1-\rho}
                -  \frac{2 \rho^2}{n} \cdot \frac{1- \rho^{n-1}}{(1-\rho)^2}
        \right\}.
    \end{equation}
    Combining this with \eqref{eq:cbu2}, we find that~\eqref{eq:lpc} holds.
\end{proof}

\subsection{Discretization}\label{ss:dmeth}

If the state space is finite, then, as discussed in Section~\ref{s:aar}, the exponent
$\lL_\Phi^p = \lL_\Phi$ is equal to the log of the spectral radius of a
valuation matrix and can therefore be obtained by
numerical linear algebra.  This leads to the following idea for handling
settings where the state space is infinite and no analytical expression for
$\lL_\Phi$ exists: discretize the state process and then proceed as for the
finite state case.  In this section we investigate whether or not this
procedure leads to a good approximation to the value of $\lL_\Phi^p$ from the
original (infinite state) model.

To answer this question, we will use the model investigated just above,
in Appendix~\ref{ss:lpc}.  This is convenient because, as shown in that
section, an analytical expression for $\lL_\Phi^p$ exists.
Existence of an analytical expression allows us to make a careful comparison
between the true solution and the approximation produced by discretization.

Our first step is to discretize the Gaussian AR(1) state process
\eqref{al:sar0} using the method of \cite{rouwenhorst1995}.  This produces a
finite Markov matrix $\Pi$ and finite state space with typical elements $x,y$.
In view of \eqref{eq:phimp}, the valuation matrix $V$ corresponding to this
discretized model is given by
\begin{equation}
    \label{eq:vij}
    V(x, y) :=  \beta \exp
            \left[
                \mu_d - \gamma \mu_c + (1-\gamma) x +
                \frac{\sigma_d^2 + (\gamma \sigma_c)^2}{2}
            \right]
    \Pi(x, y).
\end{equation}
We calculate the spectral radius $r(V)$ using linear algebra routines and,
from there, compute the associated value for the stability exponent via
\eqref{eq:rvi}.  Finally, we compare the result with the true value of
$\lL_\Phi$ obtained from the analytical expression \eqref{eq:lpc}.

Figure~\ref{f:discretization_vs_true_fig} shows this comparison when the
utility parameter $\gamma$ is set to $2.5$ and the consumption and dividend
parameters are set to the values in table~I of
\cite{bansal2004risks}.\footnote{\label{fn:mp}In particular,
$\mu_c = \mu_d = 0.0015$, $\rho=0.979$, $\sigma=0.00034$, $\sigma_c =
0.0078$, $\sigma_d = 0.035$ and $\varphi=1.0$.}  The vertical axis shows the value of
  $\lL_\Phi$.  The horizontal axis shows the level of discretization, indexed
by the number of states for $\{X_t\}$ generated at the Rouwenhorst step.  The
true value of $\lL_\Phi$ at these parameters, as calculate from
  \eqref{eq:lpc}, is $-0.0031545$.  The discrete approximation of $\lL_\Phi$
  is accurate up to six decimal places whenever the state space has more than
  6 elements.  Thus, the discrete approximation is sufficiently accurate to
    implement the test $\lL_\Phi < 0$ even
    for relatively coarse discretizations.  Moreover, as shown in the figure, the
    approximation of $\lL_\Phi$ converges to the true value as the number of
    states increases.  We experimented with other parameter values and found
    similar results.

\begin{figure}
    \centering
    \scalebox{0.7}{\includegraphics[clip=true, trim=0mm 0mm 0mm 0mm]{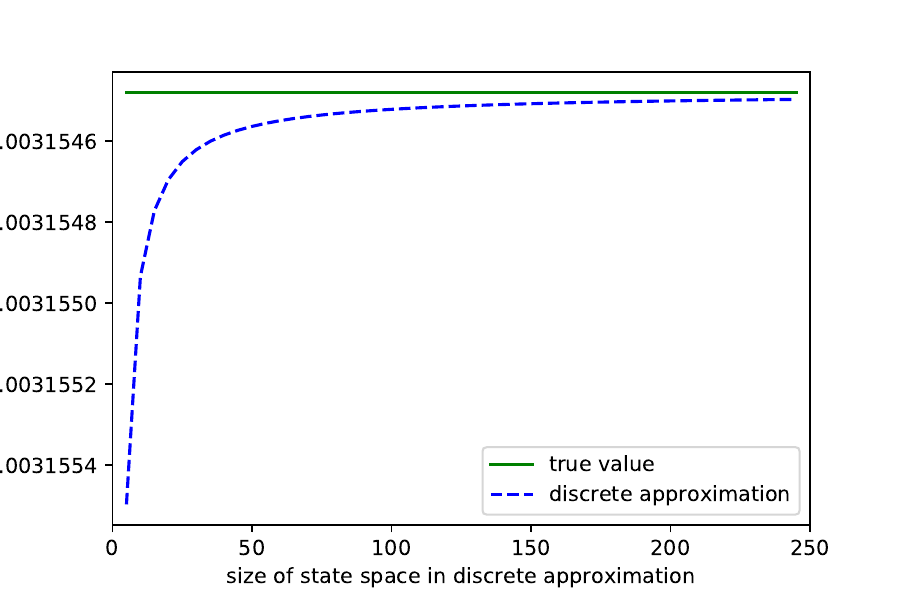}}
    \caption{\label{f:discretization_vs_true_fig} Accuracy of discrete approximation of $\lL_\Phi$}
\end{figure}

\subsection{A Monte Carlo Method}

Discretization works well for low dimensional
state processes but is susceptible to the curse of dimensionality.
For this reason,  we also propose a Monte Carlo
method that requires only the ability to simulate the SDF process
$\{\Phi_t\}$.  As well as being less susceptible to the curse of
dimensionality, this method
has the advantage that simulation of the SDF process can be targeted for
parallelization across CPUs or GPUs.

The idea behind the Monte Carlo method is to approximate $\lL_\Phi$ via
\begin{equation}
    \label{eq:mc}
    \lL_\Phi(n, m) :=
    \frac{1}{n}
    \ln
    \left\{
        \frac{1}{m} \sum_{j=1}^m  \prod_{i=1}^{n} \Phi^{(j)}_i
    \right\},
\end{equation}
where each $\Phi^{(j)}_1, \ldots, \Phi^{(j)}_n$ is an independently simulated
path of $\{\Phi_t\}$, and $n$ and $m$ are suitably chosen integers.
The idea relies on the strong law of large numbers, which yields
$\frac{1}{m} \sum_{j=1}^m \prod_{i=1}^{n} \Phi^{(j)}_i \to \EE \,
\prod_{i=1}^{n} \Phi_i $ with probability one, combined with the fact that
$Z_n \to Z$ almost surely implies $g(Z_n) \to g(Z)$ almost surely whenever $g
\colon \RR \to \RR$ is continuous.

These are asymptotic results.
Table~\ref{tab:t1} tests finite sample behavior.
We again use the constant volatility model from
Section~\ref{ss:crra}, comparing Monte Carlo approximations of $\lL_\Phi$ with
the true value obtained from in \eqref{eq:lpc}.
Consumption and dividend growth parameters are as in
footnote~\ref{fn:mp}. The true value of $\lL_\Phi$
is $-0.0031545$, as shown in the caption
for the table.  The interpretation of $n$ and $m$ in the table is consistent
with the left hand side of \eqref{eq:mc}.  For each $n, m$ pair, we compute
$\lL_\Phi(n, m)$ 1,000 times using independent draws and present the mean and
their standard error in the corresponding cell.  We find that
the Monte Carlo approximation is accurate up to
four decimal places when $n =750$ and standard deviations
are small.  At least for this model, the Monte Carlo method
can determine the sign of $\lL_\Phi$.

{\small
\begin{table}
    \centering
    \caption{Monte Carlo spectral radius estimates when $\lL_\Phi = -0.0031545$}
    \begin{tabular}{l|rrrrr}
    \toprule
        & m = 1000 & m = 2000 & m = 3000 & m = 4000 & m = 5000 \\
    \midrule
    n = 250 & -0.0033183 & -0.0032524 & -0.0032434 & -0.0032533 & -0.0032353 \\
             & (0.000003) & (0.000002) & (0.000001) & (0.000001) & (0.000001) \\
    n = 500 & -0.0032045 & -0.0032149 & -0.0031948 & -0.0031907 & -0.0031922 \\
             & (0.000002) & (0.000001) & (0.000001) & (0.000001) & (0.000001) \\
    n = 750 & -0.0031985 & -0.0031841 & -0.0031748 & -0.0031784 & -0.0031890 \\
             & (0.000002) & (0.000001) & (0.000001) & (0.000001) & (0.000001) \\
    \bottomrule
    \end{tabular}
    \label{tab:t1}
\end{table}
}

\section{Proofs}\label{s:ax}

If $\eE$ is a Banach lattice, then an
\emph{ideal} in $\eE$ is a vector subspace $L$ of $\eE$ with $x \in
L$ whenever $|x| \leq |y|$ and $y \in L$.   The \emph{spectral radius} of a
bounded linear operator $M$ from $\eE$ to itself is the supremum of
$|\lambda|$ for all $\lambda$ in the spectrum of $A$.  The operator $M$ is
called \emph{compact} if the image under $M$ of the unit ball in $\eE$ has
compact closure.  $M$ is called \emph{eventually compact} if
there exists an $i \in \NN$ such that $M^i$ is compact.  $M$ is called
\emph{positive} if it maps the positive cone of $\eE$ into itself.   A positive
linear operator $M$ is called \emph{irreducible} if the only closed ideals $J
\subset \eE$ satisfying $M(J) \subset J$ are $\{0\}$ and $\eE$.
See \cite{abramovich2002invitation} or \cite{meyer2012banach} for more
details.

If $\XX$ is an Polish space, $\pi$ is a finite Borel measure on $\XX$ and $p
\geq 1$, then $L_p(\pi) := L_1(\XX, \RR, \pi)$ denotes is the set of all Borel
measurable functions $f$ from $\XX$ to $\RR$ satisfying $\int |f|^p \diff \pi
< \infty$.  The norm on $L_p(\pi)$ is $\| f\| := (\int |f|^p \diff
\pi)^{1/p}$.  Functions equal $\pi$-almost everywhere are identified.
Convergence on $L_p(\pi)$ is with respect to the norm topology generated by
$\| \cdot \|$.  We write $f \leq g$ if $f \leq g$
pointwise $\pi$-almost everywhere, and $f \ll g$ if $f < g$ holds
pointwise $\pi$-almost everywhere.  The positive cone of $L_p(\pi)$ is all $f
\in L_p(\pi)$ with $f \geq 0$.  We denote this set by $\hH_p$, which conforms
with our previous definition (cf., Theorem~\ref{t:bk}).

\subsection{Operator Compactness in Spaces of Summable Functions}



Assumption~\ref{a:co} requires that $V$ is
eventually compact as a linear map from $L_p(\XX, \RR, \pi)$ to itself.
Here we give a sufficient condition focused on the applications
in Section~\ref{s:aar}.  Take $\XX = \RR$ and $p=2$.
  In the proposition below, $q$ is a
stochastic density kernel on $\RR^2$ with stationary density $\pi$ and two
step density kernel $q^2$.

\begin{proposition}
    \label{p:eco}
    Let $M$ be an operator that maps $f$ in $L_2(\pi)$ into
    \begin{equation}
        M f(x) = g(x) \int f(y) q(x, y) \diff y
        \qquad (x \in \RR),
    \end{equation}
    where $g$ is a measurable function from $\RR$ to $\RR_+$.  If $q$ is
    time-reversible and
    \begin{equation}
        \label{eq:bcg}
        \int g(x) q^2(x, x) \diff x < \infty,
    \end{equation}
    then $M$ is a compact linear operator on $L_2(\pi)$.\footnote{The
        statement that $q$ is \emph{time-reversible} means that $q(x, y)
        \pi(x) = q(y, x) \pi(y)$ for all $x, y \in \RR$.  A number of our
        results use the fact that $q(x, \cdot) = N(\rho x, \sigma^2)$ for some
        $\sigma > 0$ and $|\rho| < 1$ implies that $q$ is time-reversible.
    See, e.g., \cite{o2014analysis}.}
\end{proposition}

\begin{proof}
    We can express the operator $M$ as
    \begin{equation*}
        M f(x) = \int f(y) k(x, y) \pi(y) \diff y
        \quad \text{where} \quad
        k(x, y) := \frac{g(x) q(x, y)}{\pi(y)}.
    \end{equation*}
    By theorem~6.11 of \cite{weidmann2012linear}, the operator $M$ will be
    Hilbert--Schmidt in $L_2(\pi)$, and hence compact, if the kernel $k$ satisfies
    \begin{equation*}
        \int \int k(x, y)^2 \pi(x) \pi(y) \diff x \diff y < \infty.
    \end{equation*}
    Using the definition of $k$ and the time-reversibility of $q$, this translates
    to
    \begin{equation*}
        \int g(x) \int q(x, y) q(y, x) \diff y \diff x  < \infty.
    \end{equation*}
    This completes the proof because, by definition, $q^2(x, x) = \int q(x, y)
    q(y, x) \diff y$.
\end{proof}

\subsection{Remaining Proofs}

Throughout the following we impose Assumptions~\ref{a:p}--\ref{a:co}.
The symbol $p$ represents the constant in Assumption~\ref{a:co}.
As before, $\Pi$ is a stochastic kernel on $\XX$ and $\{X_t\}$ is a stationary
Markov process on $\XX$ with stochastic kernel $\Pi$ and common marginal
distribution $\pi$.\footnote{In other words, $\Pi$ is a function from $(\XX,
    \bB)$ to $[0, 1]$ such that $B \mapsto \Pi(x, B)$ is a probability measure
    on $(\XX, \bB)$ for each $x \in \XX$, and $x \mapsto \Pi(x, B)$ is
$\bB$-measurable for each $B \in \bB$.  The process $\{X_t\}$ satisfies
$\PP\{X_{t+1} \in B \given X_t = x\} = \Pi(x, B)$ for all $x$ in $\XX$ and $B
\in \bB$.}
The symbol $\EE_x$ will indicate conditioning on the event $X_0 = x$,
so that, for any $h \in L_1(\pi)$ and any $n \in \NN$, we have
\begin{equation}
    \label{eq:mp}
    \EE_x h(X_n) = \int h(x) \Pi^n(x, \diff y).
\end{equation}
Also, for convenience, we set
\begin{equation*}
    \hat g(x) := (Vg)(x) =
    \int
        \int \phi(x, y, \eta) g(x, y, \eta) \nu(\diff \eta)
    \Pi(x, \diff y).
\end{equation*}

\begin{lemma}
    \label{l:aj}
    For any $h \in \hH_p$ and all $x \in \XX$ we have
    \begin{equation}
        \label{eq:aj}
        V^n h(x)
            = \EE_x  \prod_{i=1}^n \Phi_i \, h(X_n) .
    \end{equation}
\end{lemma}

\begin{proof}
    Equation \eqref{eq:aj} holds when $n=1$ because
    \begin{equation*}
        V h(x)
         = \int  \int \phi(x, y, \eta) \nu(\diff \eta) h(y) \Pi(x, \diff y)
         = \EE_x \, \Phi_1 h(X_1).
    \end{equation*}
    Now suppose \eqref{eq:aj} holds at arbitrary $n \in \NN$. We claim it also
    holds at $n+1$.  Indeed,
    \begin{equation*}
        V^{n+1} h(x)
         = \EE_x \, \Phi_1 V^n h(X_1)
          = \EE_x \, \Phi_1 \, \EE_{X_1} \prod_{i=2}^{n+1} \Phi_i \, h(X_{n+1})
          = \EE_x \, \EE_{X_1} \, \prod_{i=1}^{n+1} \Phi_i \, h(X_{n+1}).
    \end{equation*}
    An application of the law of iterated expectations completes the proof.
\end{proof}

\begin{lemma}
    \label{l:tpr}
    For each $h \in \hH_p$, $x \in \XX$ and $n \in \NN$ we have
    \begin{equation*}
        V^n h(x) = 0 \implies \int h(y) \Pi^n(x, \diff y) = 0.
    \end{equation*}
\end{lemma}

\begin{proof}
    Fix $h \in \hH_p$, $x \in \XX$ and $n \in \NN$ with
    $V^n h(x) = 0$.  It follows from Lemma~\ref{l:aj} that
    $\EE_x  \prod_{i=1}^n \Phi_i \, h(X_n) = 0$, which in turn implies that
    $\prod_{i=1}^n \Phi_i \, h(X_n) = 0$ holds $\PP_x$-a.s.  But then, by the
    positivity in Assumption~\ref{a:p},
    $h(X_n) = 0$ holds $\PP_x$-a.s. Hence $\EE_x h(X_n) = 0$.
    By \eqref{eq:mp}, this is equivalent to $\int h(y) \Pi^n(x, \diff y) = 0$.
\end{proof}

\begin{lemma}
    \label{l:qiqi}
    If $h \in \hH_p$ with $h \gg 0$, then $V^n h \gg 0$ for all $n \in \NN$.
\end{lemma}

\begin{proof}
    It suffices to show this is true when $n=1$, after
    which we can iterate.  To this end, fix $h \in \hH_p$ with $h > 0$ on $B \in
    \bB$ with $\pi(B) = 1$.  Suppose that
    \begin{equation*}
        V h(x) = \int  h(y)
        \left[ \int \phi(x, y, \eta) \nu(\diff \eta) \right]
        \Pi(x, \diff y) = 0.
    \end{equation*}
    Since $\phi$ is positive, we must then have $\Pi(x, B) = 0$.  But $\pi$ is
    invariant, so $\pi(B) = \int \Pi(x, B) \pi(\diff x) = 0$.  Contradiction.
\end{proof}

\begin{lemma}
    \label{l:irr}
    The valuation operator $V$ is irreducible on $L_p(\pi)$.
\end{lemma}

\begin{proof}
    Suppose to the contrary that there exists a closed ideal $J$ in $L_p(\pi)$
    such that $V$ is invariant on $J$ and $J$ is neither $\emptyset$ nor $L_p(\pi)$
    itself.  Since $J$ is a closed ideal in $L_p(\pi)$,
    there exists a set $B \in
    \bB$ such that $J = \setntn{f \in L_p(\pi)}{f = 0 \text{ $\pi$-a.e. on }
    B}$.\footnote{See, for example, \cite{gerlach2012new}, p.~765.}  Moreover,
    since $J$ is neither empty nor the whole space, it must be that, for this
    set $B$ that defines $J$, we have $0 < \pi(B) < 1$.

    Because $V$ is invariant on $J$, we have $V^n h \in J$ for all $h \in J$
    and $n \in \NN$.  In particular, $V^n \1_{B^c}$ is in $J$ for all $n \in
    \NN$.  This means that $V^n \1_{B^c}(x) = 0$ for $\pi$-almost all $x \in
    B$ and all $n$ in $\NN$.  Fixing an $x \in B$ and applying
    Lemma~\ref{l:tpr}, we then have $\Pi^n(x, B^c) = 0$ for all $n \in \NN$.
    But $\pi(B) < 1$, so $\pi(B^c) > 0$.  This contradicts irreducibility of
    the stochastic kernel $\Pi$, which in turn violates Assumption~\ref{a:i}.
\end{proof}

The following is a local spectral radius result suitable for $L_p(\pi)$ that
draws on \cite{zabreiko1967bounds} and
\cite{krasnosel2012approximate}.\footnote{The result suits $L_p(\pi)$ because
    it allows the interior of the positive cone to be empty.} The proof provided
here is due to Miros\l{}awa Zima (private communication).  In the statement of
the theorem, a quasi-interior element of the positive cone of a Banach lattice
$\eE$ is a nonnegative element $h$ satisfying $\la h, g \ra > 0$ for any
nonzero element of the positive cone of the dual space $\eE^*$.  (See
\cite{krasnosel2012approximate} for more details.)

\begin{theorem}
    \label{t:lsr}
    Let $h$ be an element of a Banach lattice $\eE$ and let $M$ be a positive
    and compact linear operator.  If $h$ is quasi-interior, then
        $\| M^n h \|^{1/n} \to r(M)$ as $n \to \infty$.
\end{theorem}

\begin{proof}
    Let $h$ and $M$ be as in the statement of the theorem and let $\eE_+$ be
    the positive cone of $\eE$.
    Let $r(h, M) := \limsup_{n \to \infty} \| M^n h \|^{1/n}$.    From the
    definition of $r(M)$ it is clear that $r(h, M) \leq r(M)$.  To see that
    the reverse inequality holds,  let $\lambda$ be a constant
    satisfying $\lambda > r(h, M)$ and let
    \begin{equation}
        \label{eq:xl}
        h_\lambda := \sum_{n=0}^\infty \frac{M^n h}{\lambda^{n+1}}.
    \end{equation}
    The point $h_\lambda$ is a well-defined element of $\eE_+$ by
        $\limsup_{n \to \infty} \| M^n h \|^{1/n} < \lambda$
    and Cauchy's root test.  It is also quasi-interior,
    since the sum in \eqref{eq:xl} includes the quasi-interior element $h$, and
    since $M$ maps $\eE_+$ into itself.  Moreover, by standard Neumann
    series theory (e.g.,
    \cite{krasnosel2012approximate}, theorem~5.1), the point $h_\lambda$
    also has the representation $h_\lambda = (\lambda I - M)^{-1} h$,
    from which we obtain $\lambda h_\lambda - M h_\lambda = h$.  Because
    $h \in \eE_+$, this implies that
         $M h_\lambda \leq \lambda h_\lambda$.
    Applying this last inequality, compactness of $M$,
    quasi-interiority of $h_\lambda$ and theorem~5.5 (a) of
    \cite{krasnosel2012approximate}, we must have $r(M) \leq \lambda$.
    Since this inequality was established for an arbitrary $\lambda$
    satisfying $\lambda > r(h, M)$, we conclude that $r(h, M) \geq r(M)$.

    We have shown that $\limsup_{n \to \infty} \| M^n h \|^{1/n} = r(M)$.
    Since $M$ is compact, Corollary~1 of \cite{danevs1987local}
    gives $\limsup_{n \to \infty} \| M^n h \|^{1/n} = \lim_{n \to \infty} \| M^n
        h \|^{1/n}$.
\end{proof}

\begin{theorem}
    \label{t:lsr2}
    The growth exponent $\lL_\Phi^p$ satisfies $\exp(\lL_\Phi^p) = r(V)$,
    where $r(V)$ is the spectral radius of $V$ in $L_p(\pi)$.
\end{theorem}

\begin{proof}
    Let $\1 = \1_\XX \equiv 1$ and let $\| \cdot\|$ be the norm in $L_p(\pi)$.  By Lemma~\ref{l:aj}, we
    have $V^n \1(x) = \EE_x  \prod_{i=1}^n \Phi_i$.  Using this and the
    definition of $\lL_\Phi^p$ in \eqref{eq:prpp}, we have
    \begin{equation*}
        \exp(\lL_\Phi^p)
        = \lim_{n \to \infty}
        \left\{
            \EE \left[ \EE_x  \prod_{t=1}^n \, \Phi_t \right]^p
        \right\}^{1/(np)}
            = \lim_{n \to \infty} \frac{1}{n} \ln \| V^n \1 \|.
    \end{equation*}
    As a consequence, it suffices to show that
    \begin{equation}
        \label{eq:lsr2}
        \lim_{n \to \infty} \|  V^n \1 \|^{1/n} = r(V).
    \end{equation}
    In doing so, we aim to apply~Theorem~\ref{t:lsr}.  We cannot do so
    directly because $V$ is not compact.
    However, by Assumption~\ref{a:co} we can choose an $i \in \NN$ such that $V^i$ is a compact linear operator on $L_p(\pi)$.  Fix $j \in \NN$ with $0 \leq j \leq i-1$.
    By Lemma~\ref{l:qiqi} we know that $V^j \1$ is positive $\pi$-almost
    everywhere on $\XX$, and is therefore quasi-interior.\footnote{By the
        Riesz Representation Theorem, the dual
        space of $L_p(\pi)$ is isometrically isomorphic to $L_q(\pi)$ where
        $1/p+1/q=1$.
        If $g$ is a nonnegative and nonzero element of $L_q(\pi)$
        then it is positive on a set of positive $\pi$ measure.  Since
        $f \gg 0$ on $\XX$, the produce $fg$ must be
        positive on a set of positive $\pi$ measure.  Hence $\int fg \diff \pi > 0$,
    so $f$ is quasi-interior.} As a result,
    Theorem~\ref{t:lsr} applied to $V^i$ with initial condition $h := V^j \1$ yields
    \begin{equation*}
        \| V^{in} V^j \1 \|^{1/n}
        = \| V^{in + j} \1 \|^{1/n}
        \to r(V^i)
        \qquad (n \to \infty).
    \end{equation*}
    But $r(V^i) = r(V)^i$, so
    $\| V^{in+j} \1 \|^{1/(in)} \to r(V)$ as $n \to \infty$.
    It follows that
    \begin{equation*}
        \| V^{in+j} \1 \|^{1/(in+j)} \to r(V).
    \end{equation*}
    As this is shown to be true for any integer $j$ satisfying  $0 \leq j \leq i-1$,
    we can conclude that \eqref{eq:lsr2} is valid.
\end{proof}

To prove Theorem~\ref{t:bk}, we will also need the following two lemmas:

\begin{lemma}
    \label{l:elc}
    The equilibrium price operator $T$ is a self-map on $\hH_p$.  It has a
    fixed point in $\hH_p$ if and only if there exist elements $h_0, h$ in
    $\hH_p$ such that $T^n h_0 \to h$ as $n \to \infty$.
\end{lemma}

\begin{proof}
    To see that $T$ is a self-map on $\hH_p$, fix $h \in \hH_p$ and
    recall from~\eqref{eq:deff} that
        $T h (x)
        = Vh(x) + \EE_x \, \Phi_{t+1} G_{t+1}$.
    The fact that $V$ maps $\hH_p$ to itself, which is implied by
    Assumption~\ref{a:co}, combined with Minkowski's inequality, means
    we need only prove that the function $m(x) := \EE_x \, \Phi_{t+1}
    G_{t+1}$ is in $\hH_p$. This will be true if $m(X_t)$ has finite $p$-th
    moment under $\EE$.  By Jensen's inequality and the law of iterated
    expectations, it suffices to show that $\EE \, (\Phi_{t+1} G_{t+1})^p <
    \infty$, which is true by the moment condition in Assumption~\ref{a:co}.

    To prove the second claim in Lemma~\ref{l:elc}, we
    suppose first that there exist $h_0, h$ in $\hH_p$ such that $T^n h_0 \to h$ as $n \to \infty$.
    Since $T$ maps $f$ into $Vf + \hat g$ and $V$ is a bounded linear operator on
    $L_p(\pi)$,  we know that $T$ is continuous as a self-map on $L_p(\pi)$.
    Letting $h_n = T^n h_0$, we have $h_n \to h$ and hence, by continuity, $T
    h_n \to Th$.  But, by the definition of the sequence $\{h_n\}$, we must
    also have $T h_n \to h$.  Hence $Th=h$.

    Conversely, if $T$ has a fixed point $f \in \hH_p$, then the condition in
    the statement of Lemma~\ref{l:elc} is satisfied with $h_0=h=f$.
\end{proof}

\begin{proposition}
    \label{p:tfp}
    If $T$ has a fixed point in $\hH_p$, then $\lL_\Phi^p < 0$.
\end{proposition}

\begin{proof}
    Let $V^*$ be the adjoint operator associated with $V$.
    Since $V$ is irreducible (see Lemma~\ref{l:irr}) and $V^i$ is compact for
    some $i$, the version of the
    Krein--Rutman theorem presented in lemma~4.2.11 of \cite{meyer2012banach}
    together with the Riesz Representation Theorem imply existence of
    an $e^*$ in the dual space $L_q(\pi)$ such that
    \begin{equation}
        \label{eq:pao}
        e^* \gg 0
        \;\; \text{ and} \quad
        V^* e^* = r(V) e^*.
    \end{equation}
    Let $h$ be a fixed point of $T$ in $\hH_p$.  Clearly $h$ is nonzero, since
    $T0 = V0 + \hat g = \hat g$ and $\hat g$ is not the zero function (see
    Assumption~\ref{a:p}).
    Moreover, since $h$ is a fixed point, we have $h = Vh + \hat g$ and hence,
    with the inner production notation $\la \phi, f \ra := \int \phi f \diff
    \pi$,
    \begin{equation*}
        \la e^*, h \ra
        = \la e^*, V h \ra + \la e^*, \hat g \ra
        = \la V^* e^*, h \ra + \la e^*, \hat g \ra
        = r(V) \la e^*, h \ra + \la e^*, \hat g \ra.
    \end{equation*}
    In other words,
    \begin{equation*}
        (1 - r(V)) \la e^*, h \ra = \la e^*, \hat g \ra.
    \end{equation*}

    Both $h$ and $\hat g$ are nonzero in $L_p(\pi)$ and $e^*$ is positive
    $\pi$-a.e., so $\la e^*, h \ra > 0$ and $\la e^*, \hat g \ra > 0$.
    It follows that $r(V) < 1$.  By Theorem~\ref{t:lsr2}, we have $\lL_\Phi^p = \ln
    r(V)$, which proves the claim in the lemma.
\end{proof}

\begin{proof}[Proof of Theorem~\ref{t:bk}]
    By Lemma~\ref{l:elc}, (b) and (c) of Theorem~\ref{t:bk} are equivalent, so
    it suffices to show that (d) $\implies$ (c) $\implies$ (a)
    $\implies$ (d).  Of these, the implications (d) $\implies$ (c)
    is trivial, and (c) $\implies$ (a) was established in
    Proposition~\ref{p:tfp}.  Hence we need only show that (a) $\implies$ (d).

    To see that (a) implies (d), suppose that $\lL_\Phi^p < 0$.  Then, by
    Theorem~\ref{t:lsr2}, we have $r(V) < 1$.  Using Gelfand's
    formula for the spectral radius, which states that $r(V) =
    \lim_{n\to\infty} \| V^n \|^{1/n}$ with $\| \cdot \|$ as the operator
    norm, we can choose $n \in \NN$ such that $\| V^n \| < 1$.  Then, for any $h,
    h' \in \hH_p$ we have
    \begin{equation*}
        \| T^n h - T^n h' \|
        = \| V^n h - V^n h' \|
        = \| V^n (h -  h') \|
        \leq \| V^n \| \cdot \|  h -  h' \|.
    \end{equation*}
    Observe that $\hH_p$ is closed in $L_p(\pi)$, since $L_p(\pi)$ is a Banach lattice.
    Hence $\hH_p$ is complete in the norm topology.  Existence, uniqueness and
    global stability now follow from
    a well-known extension to the Banach contraction mapping theorem
    (see, e.g., p.~272 of \cite{Wagner1982FixedPoint}).

    Lastly, to see that \eqref{eq:fsg} holds, suppose that (a)--(d) are true.
    Then $r(V) < 1$, which implies that
    $(I - V)^{-1}$ is well-defined on $\hH_p$ and equals
    $\sum_{i=0}^\infty V^i$ (see, e.g., theorem~2.3.1 and corollary~2.3.3
    of \cite{atkinson2009theoretical}).  In particular, the fixed point of $T$
    is given by $h^* = \sum_{n=0}^\infty V^n \hat g$.  Applying
    \eqref{eq:aj} to this sum verifies the claim in \eqref{eq:fsg}.
\end{proof}

\begin{proof}[Proof of Proposition~\ref{p:fc}]
    Fix $p \geq 1$.
    If Assumptions~\ref{a:p}--\ref{a:i} hold and $\XX$ is a finite set endowed
    with the discrete topology, then all functions from $\XX$ to $\RR$ are
    measurable and have finite $p$-th moment, so
    $L_p(\XX, \RR, \pi) = \RR^\XX$ and $\hH_p = \RR_+^\XX$. It follows that $\hat g \in
    \hH_p$ and $V$ is a bounded linear operator from $L_p(\XX, \RR, \pi)$ to
    itself (since every linear operator mapping a finite dimensional normed
    vector space to itself is bounded).  By the Heine--Borel theorem, bounded
    subsets in finite dimensional space have compact closure, so $V$ is also
    (eventually) compact.  Thus, Assumption~\ref{a:co} holds.
    Finally, $\lL_\Phi^p = \lL_\Phi^1$ by the identity in \eqref{eq:rvi},
    since, in a finite dimension normed linear space, the spectral radius is
    independent of the choice of norm (due to equivalence of norms combined
    with Gelfand's formula for the spectral radius).
\end{proof}



\bibliographystyle{ecta}

\bibliography{jet_final}

\begin{thebibliography}{47}
\newcommand{\enquote}[1]{``#1''}
\expandafter\ifx\csname natexlab\endcsname\relax\def\natexlab#1{#1}\fi

\bibitem[\protect\citeauthoryear{Abel}{Abel}{1990}]{abel1990asset}
\textsc{Abel, A.~B.} (1990): \enquote{Asset prices under habit formation and
  catching up with the Joneses,} \emph{The American Economic Review}, 38--42.

\bibitem[\protect\citeauthoryear{Abramovich, Abramovich, and
  Aliprantis}{Abramovich et~al.}{2002}]{abramovich2002invitation}
\textsc{Abramovich, Y.~A., Y.~A. Abramovich, and C.~D. Aliprantis} (2002):
  \emph{An invitation to operator theory}, vol.~1, American Mathematical Soc.

\bibitem[\protect\citeauthoryear{Alvarez and Jermann}{Alvarez and
  Jermann}{2001}]{alvarez_jermann:2001}
\textsc{Alvarez, F. and U.~J. Jermann} (2001): \enquote{Quantitative Asset
  Pricing Implications of Endogenous Solvency Constraints,} \emph{Review of
  Financial Studies}, 14, 1117--1151.

\bibitem[\protect\citeauthoryear{Alvarez and Jermann}{Alvarez and
  Jermann}{2005}]{alvarez2005}
---\hspace{-.1pt}---\hspace{-.1pt}--- (2005): \enquote{Using Asset Prices to
  Measure the Persistence of the Marginal Utility of Wealth,}
  \emph{Econometrica}, 73, 1977--2016.

\bibitem[\protect\citeauthoryear{Atkinson and Han}{Atkinson and
  Han}{2009}]{atkinson2009theoretical}
\textsc{Atkinson, K. and W.~Han} (2009): \emph{Theoretical Numerical Analysis:
  A Functional Analysis Framework}, vol.~39, Springer Science \& Business
  Media.

\bibitem[\protect\citeauthoryear{Backus, Gregory, and Zin}{Backus
  et~al.}{1989}]{backus_gregory_zin:1989}
\textsc{Backus, D.~K., A.~W. Gregory, and S.~E. Zin} (1989): \enquote{Risk
  Premiums in the Term Structure: {E}vidence from Artificial Economies,}
  \emph{Journal of Monetary Economics}, 24, 371--399.

\bibitem[\protect\citeauthoryear{Bansal and Yaron}{Bansal and
  Yaron}{2004}]{bansal2004risks}
\textsc{Bansal, R. and A.~Yaron} (2004): \enquote{Risks for the long run: A
  potential resolution of asset pricing puzzles,} \emph{The Journal of
  Finance}, 59, 1481--1509.

\bibitem[\protect\citeauthoryear{Barro}{Barro}{2006}]{barro:2006}
\textsc{Barro, R.~J.} (2006): \enquote{Rare Disasters and Asset Markets in the
  Twentieth Century,} \emph{Quarterly Journal of Economics}, 121, 823--866.

\bibitem[\protect\citeauthoryear{Blanchard and Kahn}{Blanchard and
  Kahn}{1980}]{blanchard1980solution}
\textsc{Blanchard, O.~J. and C.~M. Kahn} (1980): \enquote{The solution of
  linear difference models under rational expectations,} \emph{Econometrica},
  1305--1311.

\bibitem[\protect\citeauthoryear{Borovi{\v{c}}ka, Hansen, and
  Scheinkman}{Borovi{\v{c}}ka et~al.}{2016}]{borovivcka2016misspecified}
\textsc{Borovi{\v{c}}ka, J., L.~P. Hansen, and J.~A. Scheinkman} (2016):
  \enquote{Misspecified recovery,} \emph{The Journal of Finance}, 71,
  2493--2544.

\bibitem[\protect\citeauthoryear{Borovi{\v{c}}ka and
  Stachurski}{Borovi{\v{c}}ka and Stachurski}{2020}]{borovivcka2020necessary}
\textsc{Borovi{\v{c}}ka, J. and J.~Stachurski} (2020): \enquote{Necessary and
  sufficient conditions for existence and uniqueness of recursive utilities,}
  \emph{The Journal of Finance}, 75, 1457--1493.

\bibitem[\protect\citeauthoryear{Brogueira and Sch{\"u}tze}{Brogueira and
  Sch{\"u}tze}{2017}]{brogueira2017existence}
\textsc{Brogueira, J. and F.~Sch{\"u}tze} (2017): \enquote{Existence and
  uniqueness of equilibrium in Lucas’ asset pricing model when utility is
  unbounded,} \emph{Economic Theory Bulletin}, 5, 179--190.

\bibitem[\protect\citeauthoryear{Calin, Chen, Cosimano, and Himonas}{Calin
  et~al.}{2005}]{calin2005solving}
\textsc{Calin, O.~L., Y.~Chen, T.~F. Cosimano, and A.~A. Himonas} (2005):
  \enquote{Solving asset pricing models when the price--dividend function is
  analytic,} \emph{Econometrica}, 73, 961--982.

\bibitem[\protect\citeauthoryear{Campbell and Cochrane}{Campbell and
  Cochrane}{1999}]{campbell1999force}
\textsc{Campbell, J.~Y. and J.~H. Cochrane} (1999): \enquote{By Force of Habit:
  A Consumption-Based Explanation of Aggregate Stock Market Behavior,}
  \emph{The Journal of Political Economy}, 107, 205--251.

\bibitem[\protect\citeauthoryear{Christensen}{Christensen}{2017}]{christensen2017nonparametric}
\textsc{Christensen, T.~M.} (2017): \enquote{Nonparametric stochastic discount
  factor decomposition,} \emph{Econometrica}, 85, 1501--1536.

\bibitem[\protect\citeauthoryear{Christensen}{Christensen}{2020}]{christensen2020existence}
---\hspace{-.1pt}---\hspace{-.1pt}--- (2020): \enquote{Existence and uniqueness
  of recursive utilities without boundedness,} .

\bibitem[\protect\citeauthoryear{Cogley and Sargent}{Cogley and
  Sargent}{2008}]{cogley_sargent:2008}
\textsc{Cogley, T. and T.~J. Sargent} (2008): \enquote{The Market Price of Risk
  and the Equity Premium: {A} Legacy of the {G}reat {D}epression?}
  \emph{Journal of Monetary Economics}, 55, 454--476.

\bibitem[\protect\citeauthoryear{Collin-Dufresne, Johannes, and
  Lochstoer}{Collin-Dufresne
  et~al.}{2016}]{collindufresne_johannes_lochstoer:2016}
\textsc{Collin-Dufresne, P., M.~Johannes, and L.~A. Lochstoer} (2016):
  \enquote{Parameter Learning in General Equilibrium: {T}he Asset Pricing
  Implications,} \emph{American Economic Review}, 106, 664--698.

\bibitem[\protect\citeauthoryear{Dane{\v{s}}}{Dane{\v{s}}}{1987}]{danevs1987local}
\textsc{Dane{\v{s}}, J.} (1987): \enquote{On local spectral radius,}
  \emph{{\v{C}}asopis pro p{\v{e}}stov{\'a}n{\'\i} matematiky}, 112, 177--187.

\bibitem[\protect\citeauthoryear{Duffie}{Duffie}{2001}]{duffie2001dynamic}
\textsc{Duffie, D.} (2001): \emph{Dynamic Asset Pricing Theory}, Princeton
  University Press.

\bibitem[\protect\citeauthoryear{Forster and Nagy}{Forster and
  Nagy}{1991}]{forster1991local}
\textsc{Forster, K.-H. and B.~Nagy} (1991): \enquote{On the local spectral
  radius of a nonnegative element with respect to an irreducible operator,}
  \emph{Acta Universitatis Szegediensis}, 55, 155--166.

\bibitem[\protect\citeauthoryear{Gerlach and Nittka}{Gerlach and
  Nittka}{2012}]{gerlach2012new}
\textsc{Gerlach, M. and R.~Nittka} (2012): \enquote{A new proof of Doob's
  theorem,} \emph{Journal of Mathematical Analysis and Applications}, 388,
  763--774.

\bibitem[\protect\citeauthoryear{Hansen}{Hansen}{2012}]{hansen2012dynamic}
\textsc{Hansen, L.~P.} (2012): \enquote{Dynamic valuation decomposition within
  stochastic economies,} \emph{Econometrica}, 80, 911--967.

\bibitem[\protect\citeauthoryear{Hansen and Richard}{Hansen and
  Richard}{1987}]{hansen_richard:1987}
\textsc{Hansen, L.~P. and S.~F. Richard} (1987): \enquote{The Role of
  Conditioning Information in Deducing Testable Restrictions Implied by Dynamic
  Asset Pricing Models,} \emph{Econometrica}, 55, 587--613.

\bibitem[\protect\citeauthoryear{Hansen and Scheinkman}{Hansen and
  Scheinkman}{2009}]{hansen2009long}
\textsc{Hansen, L.~P. and J.~A. Scheinkman} (2009): \enquote{Long-Term Risk: An
  Operator Approach,} \emph{Econometrica}, 77, 177--234.

\bibitem[\protect\citeauthoryear{Knill}{Knill}{1992}]{knill1992positive}
\textsc{Knill, O.} (1992): \enquote{Positive Lyapunov exponents for a dense set
  of bounded measurable SL (2, {$\mathbb{R}$})-cocycles,} \emph{Ergodic Theory
  and Dynamical Systems}, 12, 319--331.

\bibitem[\protect\citeauthoryear{Kocherlakota}{Kocherlakota}{1990}]{kocherlakota:1990}
\textsc{Kocherlakota, N.} (1990): \enquote{On Tests of Representative Consumer
  Asset Pricing Models,} \emph{Journal of Monetary Economics}, 26, 285--304.

\bibitem[\protect\citeauthoryear{Krasnosel'skii, Vainikko, Zabreyko, Ruticki,
  and Stet'senko}{Krasnosel'skii et~al.}{2012}]{krasnosel2012approximate}
\textsc{Krasnosel'skii, M., G.~Vainikko, R.~Zabreyko, Y.~Ruticki, and
  V.~Stet'senko} (2012): \emph{Approximate Solution of Operator Equations},
  Springer Netherlands.

\bibitem[\protect\citeauthoryear{Kreps}{Kreps}{1981}]{kreps1981arbitrage}
\textsc{Kreps, D.~M.} (1981): \enquote{Arbitrage and equilibrium in economies
  with infinitely many commodities,} \emph{Journal of Mathematical Economics},
  8, 15--35.

\bibitem[\protect\citeauthoryear{Lorenz, Schmedders, and Schumacher}{Lorenz
  et~al.}{2020}]{lorenz2020nonlinear}
\textsc{Lorenz, F., K.~Schmedders, and M.~Schumacher} (2020):
  \enquote{Nonlinear Dynamics in Conditional Volatility,} Tech. rep., SSRN
  3575458.

\bibitem[\protect\citeauthoryear{Lucas}{Lucas}{1978}]{lucas1978asset}
\textsc{Lucas, R.~E.} (1978): \enquote{Asset prices in an exchange economy,}
  \emph{Econometrica}, 1429--1445.

\bibitem[\protect\citeauthoryear{Martin and Ross}{Martin and
  Ross}{2019}]{martin_ross:2019}
\textsc{Martin, I.~W. and S.~A. Ross} (2019): \enquote{Notes on the Yield
  Curve,} \emph{Journal of Financial Economics}, 133, 689--702.

\bibitem[\protect\citeauthoryear{Mehra and Prescott}{Mehra and
  Prescott}{1985}]{mehra1985equity}
\textsc{Mehra, R. and E.~C. Prescott} (1985): \enquote{The equity premium: A
  puzzle,} \emph{Journal of Monetary Economics}, 15, 145--161.

\bibitem[\protect\citeauthoryear{Meyer-Nieberg}{Meyer-Nieberg}{2012}]{meyer2012banach}
\textsc{Meyer-Nieberg, P.} (2012): \emph{Banach lattices}, Springer Science \&
  Business Media.

\bibitem[\protect\citeauthoryear{Meyn and Tweedie}{Meyn and
  Tweedie}{2009}]{meyn2009markov}
\textsc{Meyn, S. and R.~L. Tweedie} (2009): \emph{Markov chains and stochastic
  stability}, Cambridge University Press.

\bibitem[\protect\citeauthoryear{O'Donnell}{O'Donnell}{2014}]{o2014analysis}
\textsc{O'Donnell, R.} (2014): \emph{Analysis of boolean functions}, Cambridge
  University Press.

\bibitem[\protect\citeauthoryear{Pohl, Schmedders, and Wilms}{Pohl
  et~al.}{2018}]{pohl2018higher}
\textsc{Pohl, W., K.~Schmedders, and O.~Wilms} (2018): \enquote{Higher Order
  Effects in Asset Pricing Models with Long-Run Risks,} \emph{The Journal of
  Finance}, 73, 1061--1111.

\bibitem[\protect\citeauthoryear{Pohl, Schmedders, and Wilms}{Pohl
  et~al.}{2019}]{pohl2019rel}
---\hspace{-.1pt}---\hspace{-.1pt}--- (2019): \enquote{Relative Existence for
  Recursive Utility,} Tech. rep., SSRN Working paper 3432469.

\bibitem[\protect\citeauthoryear{Qin and Linetsky}{Qin and
  Linetsky}{2017}]{qin2017long}
\textsc{Qin, L. and V.~Linetsky} (2017): \enquote{Long-Term Risk: A Martingale
  Approach,} \emph{Econometrica}, 85, 299--312.

\bibitem[\protect\citeauthoryear{Rouwenhorst}{Rouwenhorst}{1995}]{rouwenhorst1995}
\textsc{Rouwenhorst, K.~G.} (1995): \enquote{Asset pricing implications of
  equilibrium business cycle models,} in \emph{Frontiers of Business Cycle
  Research}, Princeton University Press, 294--330.

\bibitem[\protect\citeauthoryear{Schorfheide, Song, and Yaron}{Schorfheide
  et~al.}{2018}]{schorfheide2018identifying}
\textsc{Schorfheide, F., D.~Song, and A.~Yaron} (2018): \enquote{Identifying
  long-run risks: A Bayesian mixed-frequency approach,} \emph{Econometrica},
  86, 617--654.

\bibitem[\protect\citeauthoryear{van Binsbergen, Hueskes, Koijen, and
  Vrugt}{van Binsbergen et~al.}{2013}]{binsbergen_hueskes_koijen_vrugt:2013}
\textsc{van Binsbergen, J., W.~Hueskes, R.~S. Koijen, and E.~B. Vrugt} (2013):
  \enquote{Equity Yields,} \emph{Journal of Financial Economics}, 110,
  503--519.

\bibitem[\protect\citeauthoryear{van Binsbergen, Brandt, and Koijen}{van
  Binsbergen et~al.}{2012}]{binsbergen_brandt_koijen:2012}
\textsc{van Binsbergen, J. H.~v., M.~W. Brandt, and R.~S.~J. Koijen} (2012):
  \enquote{On the Timing and Pricing of Dividends,} \emph{American Economic
  Review}, 102, 1596--1618.

\bibitem[\protect\citeauthoryear{Wagner}{Wagner}{1982}]{Wagner1982FixedPoint}
\textsc{Wagner, C.~H.} (1982): \enquote{A Generic Approach to Iterative
  Methods,} \emph{Mathematics Magazine}, 55, 259--273.

\bibitem[\protect\citeauthoryear{Weidmann}{Weidmann}{2012}]{weidmann2012linear}
\textsc{Weidmann, J.} (2012): \emph{Linear operators in Hilbert spaces},
  vol.~68, Springer Science \& Business Media.

\bibitem[\protect\citeauthoryear{Weil}{Weil}{1989}]{weil:1989}
\textsc{Weil, P.} (1989): \enquote{The Equity Premium Puzzle and the Risk-Free
  Rate Puzzle,} \emph{Journal of Monetary Economics}, 24, 401--421.

\bibitem[\protect\citeauthoryear{Zabreiko, Krasnosel'skii, and
  Stetsenko}{Zabreiko et~al.}{1967}]{zabreiko1967bounds}
\textsc{Zabreiko, P., M.~Krasnosel'skii, and V.~Y. Stetsenko} (1967):
  \enquote{Bounds for the spectral radius of positive operators,}
  \emph{Mathematical Notes}, 1, 306--310.

\end{thebibliography}

\end{document}